\documentclass[aps,prb,twocolumn,amsfonts,showpacs,superscriptaddress]{revtex4-1}
\usepackage{verbatim,epsfig,amsmath,amssymb,bm,epsf,graphicx,psfrag,bbold,amsthm,amsfonts}
\usepackage[bottom]{footmisc}
\usepackage{hyperref}
\usepackage{enumitem}
\usepackage{framed}
\setlist{nosep}
\usepackage[all]{xy}
\usepackage{color}
\usepackage[utf8]{inputenc}
\usepackage{natbib}
\usepackage{tikz-cd}
\usepackage{verbatim}
\usepackage{leftidx}
\newtheorem{theorem}{Theorem}
\newtheorem{lemma}{Lemma}
\newtheorem{corollary}{Corollary}

\newcommand{\ket}[1]{\mbox{$ | #1 \rangle $}}

\newcommand\bZ {{\mathbb Z}}

\newcommand\beq {\begin{equation}}
\newcommand\eeq {\end{equation}}
\newcommand\beqa {\begin{equation}\begin{array}}
\newcommand\eeqa {\end{array}\end{equation}}
\newcommand\bal {\begin{align}}
\newcommand\eal {\end{align}}
\newcommand{\bea}{\begin{eqnarray}}
\newcommand{\eea}{\end{eqnarray}}

\newcommand{\innerproduct}[2]{\langle #1| #2\rangle}

\newcommand{\ztzt}{\mathbb{Z}_2 \times \mathbb{Z}_2}
\newcommand{\ztwo}{\mathbb{Z}_2}

\newcommand{\hgc}{H^2(G,U(1))}
\newcommand{\hgcinput}[1]{H^2(#1,U(1))}

\theoremstyle{plain}

\theoremstyle{definition}

\theoremstyle{remark}

\begin{document}

\title{An elementary proof of 1D LSM theorems}

\author{Abhishodh Prakash}
\email{abhishodh.prakash@icts.res.in}
\affiliation{International Centre for Theoretical Sciences (ICTS-TIFR),
Tata Institute of Fundamental Research,
Shivakote, Hesaraghatta Hobli,
Bengaluru 560089, India}
\date{\today}

\begin{abstract}
The Lieb-Schultz-Mattis (LSM) theorem and its generalizations forbids the existence of a unique gapped ground state in the presence of certain lattice and internal symmetries and thus imposes powerful constraints on the low energy properties of quantum many-body systems. We provide an elementary proof of a class of generalized LSM theorems in 1D using matrix product state representations and the representation theory of groups. 
\end{abstract}

\maketitle
 
\section{Introduction}
One of the central goals of condensed matter physics is to determine the collective properties of a large number interacting particles that make up various materials of interest. Of particular interest are \emph{quantum materials} where the particle interactions are dominated by the laws of quantum mechanics. Recall that to specify a quantum system, we need to specify two classes of information:
\begin{enumerate}
	\item \emph{Kinematics}: Spatial dimensions, Hilbert space, particle statistics and global symmetries.
	\item \emph{Dynamics}: The Hamiltonian which is local and invariant under symmetries.
\end{enumerate}
The properties of quantum materials are encoded in the low energy properties of its Hamiltonian i.e. the ground state and `light' excitations. However, determining these is generally a very hard task.  This is because the Hilbert space of the system grows \emph{exponentially} with the number of particles and in thermodynamic limit, as the number of particles approaches infinity, the problem becomes intractable unless the Hamiltonian has special properties.

 An example of such a special case is when we consider \emph{non-interacting} electrons which is a fairly accurate description of several metals and insulators. In this case, the Hilbert space we consider is that of fermions living on sites of some lattice whose dynamics is governed by the following \emph{tight-binding} Hamiltonian:
\begin{equation}
\label{eq:tight binding}
H_{tb} = - \sum_{i,j,\alpha,\beta} t^{i j}_{\alpha \beta}~ c^\dagger_{i \alpha} c_{j \beta} + h.c.
\end{equation}
Here, $c^\dagger_{i \alpha}$ and $ c_{i \alpha}$ are the creation and annihilation operators for fermions at site $i =1 \ldots L$ and of flavor $\alpha = 1\ldots F$, $L$ is the number of sites on the lattice and $F$ is the number of flavours per site (eg: orbitals or spin). $H_{tb}$ is a hermitian matrix of size $F^L \times F^L$ that grows exponentially with system size and hence seemingly hard to analyze. However, it is easily solved by simply diagonalizing the \emph{first quantized} Hamiltonian $t^{ij}_{\alpha\beta}$ whose size, $L F \times L F$ grows quadratically with $L$ and is vastly simpler to diagonalize than the original $H_{tb}$. If present, translation invariance simplifies the problem even further and reduces it to diagonalizing a system size independent family of \emph{Bloch Hamiltonians}, $H(\vec{k})$ parametrized by \emph{crystal momenta}, $\vec{k}$ whose energies, called \emph{bands} allows us to determine if the Hamiltonian~(\ref{eq:tight binding}) is describing a metal or insulator depending on whether the ground state of $H_{tb}$ corresponds to partially or fully filled energy bands. 

This simplification is lost if we are studying \emph{strongly coupled} quantum matter where interactions between the electrons cannot be ignored. A very important example is the so-called \emph{Hubbard  model} which describes spinful electrons on a lattice with on-site repulsion  and is believed to be a good candidate to describe phenomena ranging from magnetism to high-temperature superconductivity.
\begin{equation}
H_{Hub} =  -t \sum_{\langle i,j \rangle, \sigma}  c^\dagger_{i,\sigma} c_{j,\sigma} + h.c + U \sum_{i} c^\dagger_{i,\uparrow} c_{i,\uparrow} c^\dagger_{i,\downarrow} c_{i,\downarrow}
\end{equation} 
For large $U >> t$ the low-energy physics is approximated by the Heisenberg quantum antiferromagnet (QAF)
\begin{equation}
\label{eq:QAF}
H_{QAF} = J \sum_{\langle i,j \rangle} \vec{S}_i.\vec{S}_j 
\end{equation}
where, $S^m$ are the spin-half generators of the angular momentum algebra. Obtaining the low-energy properties of (\ref{eq:QAF}) and other such spin models is a problem of great interest and in 1961, Lieb, Schultz and Mattis (LSM) were concerned with the case of the one-dimensional version of the spin-half QAF~(\ref{eq:QAF})
\begin{equation}
\label{eq:AFM}
H = J \sum_i \vec{S}_i.\vec{S}_{i+1}.
\end{equation}
In a very important work~\cite{LSM_961407}, they were able to prove that the spectrum of eq~(\ref{eq:AFM}) was gapless. Soon after, Majumdar and Ghosh~\cite{MajumdarGhoshModel} were able to show that with the following next-nearest-neighbor coupling to the spin half QAF,
\begin{equation}
\label{eq:MG_model}
H = J \sum_i \vec{S}_i.\vec{S}_{i+1} + \frac{J}{2} \sum_i \vec{S}_i.\vec{S}_{i+2},
\end{equation}
the spectrum now becomes gapped but with a two-fold ground state degeneracy due to spontaneous symmetry breaking. Finally, in a vast generalization of LSM's results, Haldane~\cite{Haldane2016ground,Haldane_PhysRevLett.50.1153} argued that any half-odd-integer spin QAF is gapless whereas the integer spin QAF was gapped. All these results and their generalizations~\cite{Hastings_2007_MPS,Oshikawa_LSM_PhysRevLett.84.1535} state that when certain lattice symmetries and half-odd-integer spin rotation symmetries are present, the ground state can never be unique and gapped but can be gapless or degenerate (either due to symmetry breaking or topological ordering). 

In a seemingly unrelated study, it was found that there existed symmetry broken phases with `incompatible' symmetries separated by a Landau forbidden \emph{deconfined quantum critical} (DQC) phase transition~\cite{Senthil_DQC1490}. More recently, these types of phase transitions were also found to occur on the boundaries of certain exotic phases called \emph{symmetry protected topological} (SPT) phases~\cite{VishwanathSenthil_PhysRevX.3.011016}. 

The modern synthesis of all these works is a novel way to understand these exotic phenomena. Instead of focusing the Hamiltonian which governs the \emph{dynamics} of these system, it is now understood that DQC transitions, boundary phenomena of SPT phases as well as LSM theorems are the result of the \emph{kinematics} of the systems- the symmetries are realized in a manner that is \emph{anomalous}~\cite{ElseThorngren_AnomalyLSM2019topological,Cho_AnomalyLSM_PhysRevB.96.195105,Metlitski_Thorngren_DQCAnomalies_PhysRevB.98.085140}. This anomaly obstructs the possibility of a trivial phase with a unique gapped ground state as well as facilitate Landau forbidden phase transitions.

In this paper, we provide an elementary proof of LSM theorems in 1D using only kinematics by exploiting the powerful technology of matrix product state representations. This gives us a clear setting where the notion of anomalous symmetries can be concretely understood. The paper is organized as follows: in sec~(\ref{sec:MPS}), we review the essential features of matrix product states with a particular focus on the implementation of symmetries, in sec~(\ref{sec:SO3_LSM}) we prove the theorem for the specific case of on-site spin-half rotation symmetry and in sec~(\ref{sec:G_LSM}), we prove the general theorem. Finally, in sec~(\ref{sec:DQC}), we put these results in perspective with regard to known models and phase diagrams. Intermediary lemmas as well as additional mathematical details are relegated to the appendices. 

\section{Matrix product state representations of quantum ground states}
\label{sec:MPS}
Let $H$ be a quantum Hamiltonian acting on the Hilbert space of a one-dimensional spin chain of $d$-level spins (qudits). Let $\ket{\psi}$ be its ground state which we can expand in some local basis as
\begin{equation}
    \ket{\psi} = \sum_{m_1, m_2, \ldots m_L}  c_{m_1 m_2 \ldots m_L} \ket{m_1,m_2,\ldots m_L}.
\end{equation}
Determining $\ket{\psi}$ typically requires specifying $d^L$ i.e. an exponentially large number of coefficients $c_{m_1 m_2 \ldots m_L}$. For one-dimensional gapped spin chains however, there exist rigorous results that prove that $\ket{\psi}$ can be efficiently represented as a \emph{matrix product state}  (MPS)~\cite{Hastings_2007_MPS,Garcia_2007_MPS}
\begin{equation}
     c_{m_1 m_2 \ldots m_L} = Tr[A^{m_1}_1A^{m_2}_2 \ldots A^{m_L}_L] 
     \label{eq:MPS}
\end{equation}
$A^{m_i}_i$ are finite matrices whose dimensions does not scale with system size. The MPS can be visualized as fig~(\ref{fig:MPS}.
\begin{figure}[!htbp]
	\centering
	\includegraphics[width=0.35\textwidth]{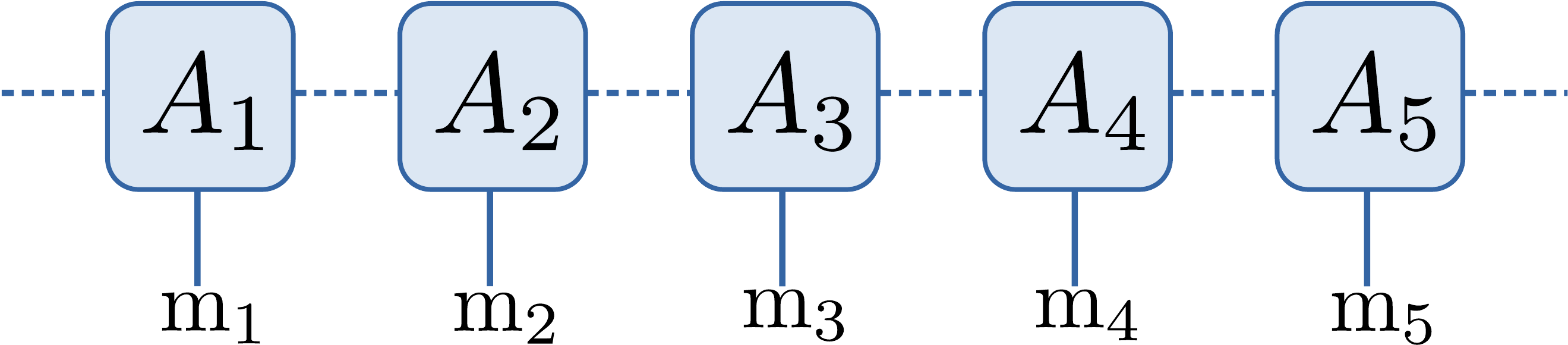}
	\caption{Graphical representation of MPS}
	\label{fig:MPS}
\end{figure}
In this representation, several features one 1D gapped systems like area-law scaling of entanglement entropy and finiteness of correlation length becomes manifest. MPS representations are at the heart of powerful numerical techniques like DMRG~\cite{White_DMRG_PhysRevLett.69.2863} and TEBD~\cite{Vidal_TEBD_PhysRevLett.91.147902} and have also facilitated several analytical breakthroughs including a \emph{complete classification} of gapped phases in 1D~\cite{FidkowskiKitaev_1dClassification_PhysRevB.83.075103,ChenGuWen_1dClassification_PhysRevB.84.235128}. 

Another advantage of MPS representations is that global symmetries can be easily incorporated and implemented. If $\ket{\psi}$ is left invariant under the action of global symmetries, this translates into conditions that the matrices $A^{m_i}_i$ have to satisfy. Let us consider two classes of global symmetries that would be relevant to this paper. 

First, if the state $\ket{\psi}$ is invariant under lattice translations  $\ket{m_i} \mapsto \ket{m_{i+1}}$, it can be shown that the matrices can be made site independent i.e. $A^{m_i}_i \equiv A^{m_i}$ so $\ket{\psi}$ can be written in a form that is manifestly translation invariant~\cite{Garcia_2007_MPS} as shown in fig~(\ref{fig:MPS_translation})
\begin{figure}[!htbp]
	\centering
	\includegraphics[width=0.35\textwidth]{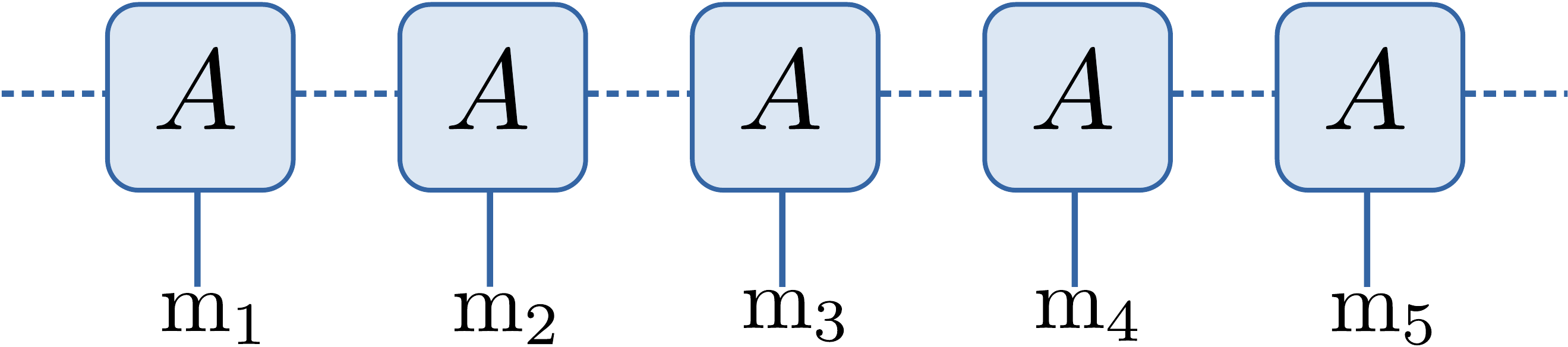}
	\caption{Translation invariant MPS}
	\label{fig:MPS_translation}
\end{figure}
\begin{equation}
    \label{eq:translation_invariant_MPS}
    \ket{\psi} = \sum_{m_1, \ldots m_L} Tr[A^{m_1} A^{m_2} \ldots A^{m_L}]  \ket{m_1,m_2,\ldots m_L}
\end{equation}
Secondly, if the state is invariant under the action of an on-site symmetry i.e. a unitary rotation of each spin by the same operator of the form $U(g) = \bigotimes_{i=1}^L D(g)$, the invariance condition $U(g) \ket{\psi} = \ket{\psi}$ can be imposed on the matrices $A^{m_i}_i$~\cite{FidkowskiKitaev_1dClassification_PhysRevB.83.075103,ChenGuWen_1dClassification_PhysRevB.84.235128} as follows (and visualized as shown in fig~(\ref{fig:MPS_internal}))
\begin{figure}[!htbp]
	\centering
	\includegraphics[width=0.35\textwidth]{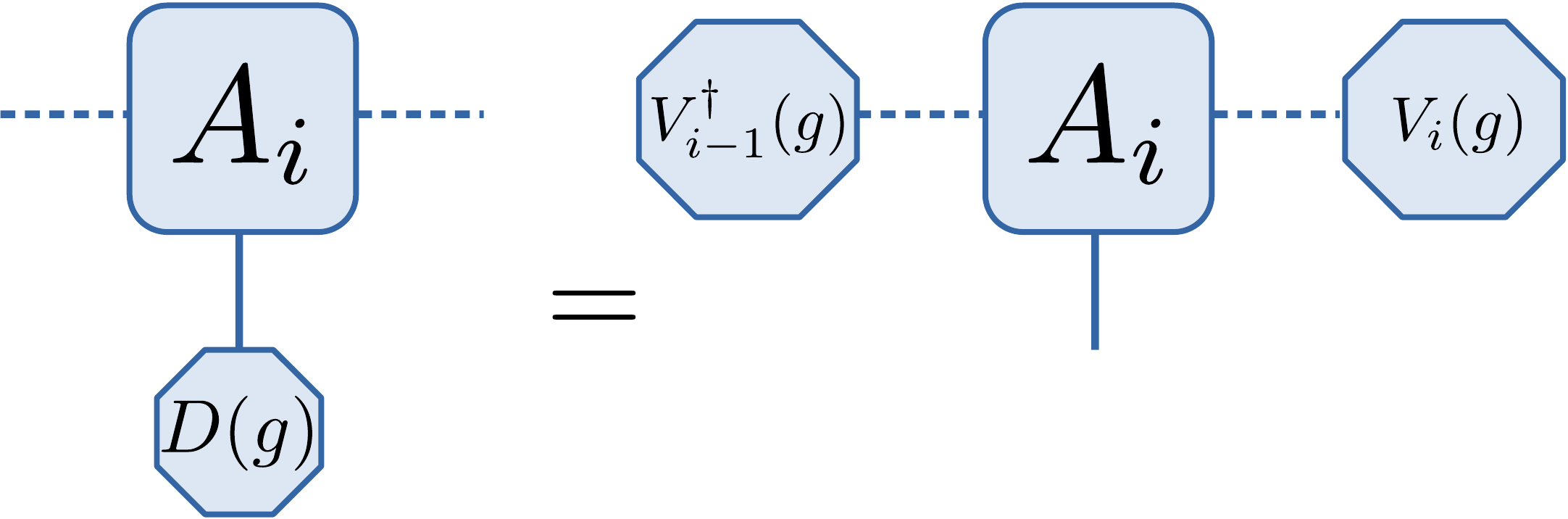}
	\caption{MPS invariant under on-site symmetry}
	\label{fig:MPS_internal}
\end{figure}
\begin{equation}
    \label{eq:MPS_onsite_invariane}
    \sum_{n} D(g)_{m n} A^{n}_{i} = V^\dagger_{i-1}(g) A_i^{m} V_i(g)
\end{equation}
In the presence of both on-site and lattice translation symmetries, when the MPS matrices are site-independent, the condition of eq~(\ref{eq:MPS_onsite_invariane}) takes on a simpler form as shown in fig~(\ref{fig:MPS_translation_internal})
\begin{figure}[!htbp]
	\centering
	\includegraphics[width=0.35\textwidth]{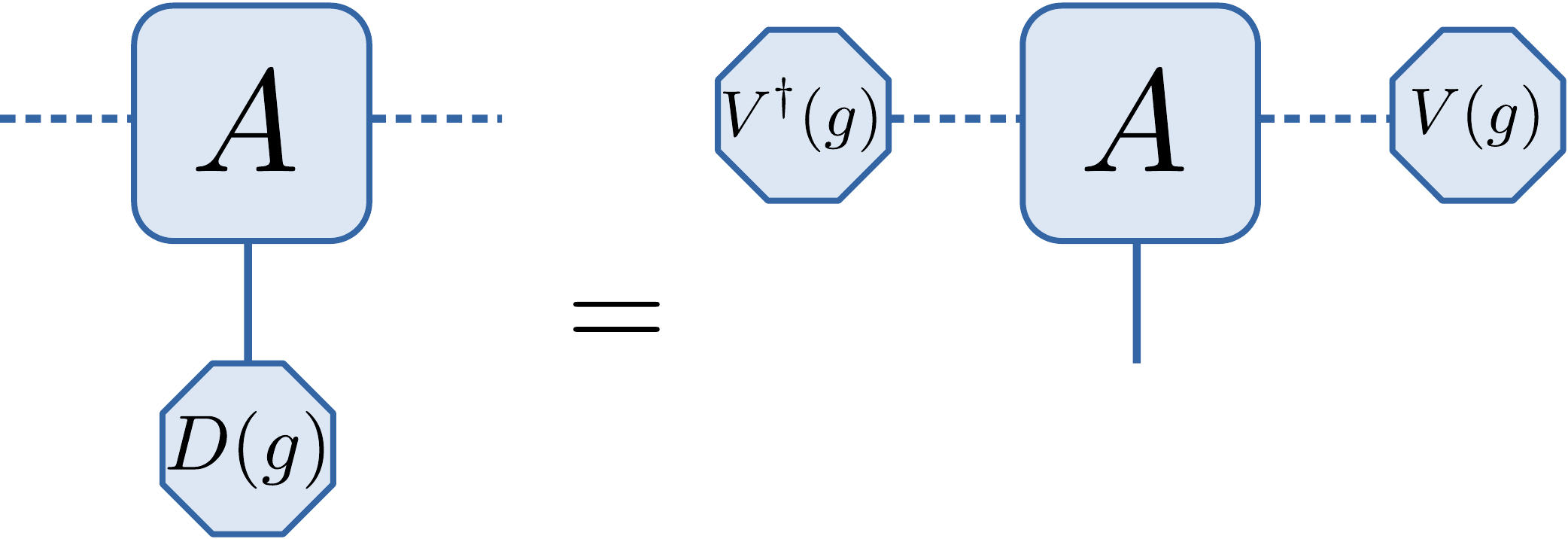}
	\caption{Translation invariant MPS with on-site symmetry}
	\label{fig:MPS_translation_internal}
\end{figure}
\begin{equation}
    \label{eq:MPS_onsite_translation_invariance}
    \sum_{n} D(g)_{m n} A^{n} = V^\dagger(g) A^{m} V(g)
\end{equation}
Finally, if the ground state is unique, the matrices also satisfy the so-called \emph{injectivity} condition. This rather technical condition is defined as follows: for any matrix $M$ acting on the \emph{bond} Hilbert space i.e. the Hilbert space in which the matrices the $A$ live (sometimes also referred to as virtual Hilbert space), there exists a finite integer value $I^*$ such that for $I>I^*$, the map $\Gamma_I$ defined below is \emph{injective}
\begin{equation}
    \Gamma_I(M) = \sum_{m_1, \ldots m_I} Tr[M A^{m_1} \ldots A^{m_I}] \ket{m_1, \ldots m_I}
\end{equation}
Put simply, this condition means that there exists some finite value $I^*$ such that for $I>I^*$, 
\begin{equation}
    \label{eq:injectivity}
	M_1 \neq M_2 \implies \Gamma_I(M_1) \neq \Gamma_I(M_2) 
\end{equation}
i.e. there are no two distinct matrices $M_1 \neq M_2$ such that $\Gamma_I(M_1) = \Gamma_I(M_2)$. Another consequence of injectivity is that the so-called transfer matrix, 
\begin{equation}
T= \sum_m A^m \otimes (A^{m})^*
\end{equation}
has a unique largest eigenvalue, which, if set to 1, ensures normalization of $\ket{\psi}$ in the thermodynamic limit. We will assume throughout this paper that all injective MPS are normalized in this manner. 

In the following sections, we will see that 1D LSM theorems follow from the mutual incompatibility of the conditions imposed by on-site symmetry and lattice translations on $A^m$ matrices with that of injectivity.

\section{The proof for the SO(3) invariant case}
\label{sec:SO3_LSM}
Let us first look at the canonical setting where the LSM theorem was originally proved. This sets the language using the representation theory of $SO(3)$ which is well known from the quantum theory of angular momentum~\cite{sakurai_napolitano_2017} which we will eventually generalize to arbitrary symmetry groups.

Hamiltonians (\ref{eq:AFM}) and (\ref{eq:MG_model}) have lattice translation symmetry as well as spin-rotation symmetry, $SO(3)$ generated by $U(g) = \bigotimes_{i=1}^L D(g)$ where $D(g) = \exp{ (i\hat{n}.\vec{\sigma} \frac{\theta}{2})}$ are the spin-half rotation matrices. In this setting, we can prove the following theorem 
\begin{theorem}
\label{th:SO(3)LSM}
A spin chain with lattice translation invariance and half-odd-integer spin rotation symmetry cannot have a unique gapped ground state.
\end{theorem}
\begin{proof}
We prove this by contradiction. Let the $J$ be the total angular momentum of the physical spin that takes some half-odd-integer value. The spin has $2J+1$ internal states which are labeled by the azimuthal quantum number  $m=-J,\ldots,+J$. Let us now assume that there does exist a unique gapped ground state. This would mean that this ground state can be expressed as an MPS (\ref{eq:MPS}) invariant under translation and on-site symmetries i.e.
\begin{equation}
    \label{eq:SO(3)_MPS_onsite_translation_invariance}
    \sum_{n = -J}^{+J} D(g)_{mn} A^{n} = V^\dagger(g) A^{m} V(g)
\end{equation}
and satisfying the injectivity condition of eq~(\ref{eq:injectivity}). $D(g)$ is the $(2J+1) \times (2J+1)$ dimensional spin-J irreducible representation (irrep) of  $SU(2)$. We can choose an appropriate basis such that $V(g)$ are also block diagonalized into different angular momentum sectors and examine the matrix elements of $A^m$ in this basis. Eq~(\ref{eq:SO(3)_MPS_onsite_translation_invariance}) tells us that $A^m$ is a \emph{tensor operator} and thus its matrix elements are constrained by the \emph{Wigner-Eckart theorem}~\cite{sakurai_napolitano_2017} as
\begin{multline}
	\label{eq:Wigner-Eckart}
    \innerproduct{j_\alpha,m_\alpha,d_\alpha|A^m}{j_\beta,m_\beta,d_\beta} \\= \innerproduct{J,m;j_\beta,m_\beta}{j_\alpha,m_\alpha} \innerproduct{j_\alpha,d_\alpha||A^m|}{j_\beta,d_\beta}
\end{multline}
where $\innerproduct{J,m;j_\beta,m_\beta}{j_\alpha,m_\alpha}$ are the Clebsch-Gordan (CG) coefficients corresponding to the fusion of angular momenta $J \otimes j_\beta \rightarrow j_\alpha$ and $\innerproduct{j_\alpha,d_\alpha||A^m|}{j_\beta,d_\beta}$ are undetermined just from symmetry but are independent of azimuthal quantum numbers. In general, the bond Hilbert space can contain multiple copies of same irrep which is labeled by $d_\alpha$ and $d_\beta$. 

Injectivity (\ref{eq:injectivity}) constrains the bond angular momenta, $j_\alpha$ and $j_\beta$ to either be only integer or only half-odd-integer valued but not both. For ease of reading, we prove this in appendix~(\ref{app:proofs}) (see lemma~(\ref{lemma:SO(3)_Bond_hilbert_space})). The proof of the theorem now follows easily. Since the physical spin $J$ is half-odd-integer, the rules of addition of angular momenta~\cite{sakurai_napolitano_2017},
\begin{equation}
    J \otimes j_\beta = |J-j_\beta|  \oplus |J-j_\beta|+1 \oplus \ldots  |J+j_\beta|-1 \oplus |J+j_\beta| ,
\end{equation}
tells us that there are no fusion channels $J \otimes j_\beta \rightarrow j_\alpha$ for $j_\alpha$ and $j_\beta$ both being either integer or half-odd integer. 

For example, consider physical spin $J=\frac{1}{2}$
\begin{equation}
    \frac{1}{2} \otimes j_\beta =\left| j_\beta -\frac{1}{2} \right| \oplus \left|j_\beta +\frac{1}{2} \right|
\end{equation}
it is clear that if $j_\beta$ is integer, the RHS above contains half-odd-integers and if $j_\beta$ is half-odd-integer, the RHS contains only integers. This means that all CG coefficients, $\innerproduct{J,m;j_\beta,m_\beta}{j_\alpha,m_\alpha}$ vanish as do the matrices $A^m$ and finally the state $\ket{\psi}$ itself! This means that our assumption of the existence of a unique symmetric ground state that can be described by an injective MPS was wrong. 

On the other hand, if the physical spin $J$ was integer, fusion channels $J \otimes j_\beta \rightarrow j_\alpha$ exist for both $j_\alpha$ and $j_\beta$ being integers or half-odd-integer and there is nothing obstructing the existence of a fully symmetric gapped ground state. 
\end{proof}

We just proved that a spin chain with lattice translation invariance and half-odd-integer spin rotation symmetry cannot have a unique gapped ground state. We now list possible alternatives for the low-energy properties.

\begin{corollary}
\label{cor:SO(3)LSM}
A spin chain with lattice translation invariance and half-odd-integer spin rotation symmetry can be either gapless or evenly degenerate with spontaneously broken translation symmetry.
\end{corollary}
\begin{proof}
	 It was shown in theorem~(\ref{th:SO(3)LSM}) that lattice translations and spin rotations are not compatible with the existence of an injective MPS. One possibility is that the ground state simply does not have an MPS description meaning the system is gapless as in the case of the QAF of eq~(\ref{eq:AFM}). Another is that the system has a gapped spectrum but the ground state is not unique. This is possible when symmetry is spontaneously broken. Here, an MPS description for the vectors of the ground space is possible but they are generally not injective. Since the Coleman-Mermin-Wagner theorem~\cite{Coleman1973,MerminWagner_PhysRevLett.17.1133} forbids spontaneous breaking of spin rotation symmetry, the only other allowed possibility is translation symmetry breaking as in the case of the MG model of eq~(\ref{eq:MG_model}). 
	
    Let us consider the spontaneous breaking of single-site translation ($\mathbb{Z}$) to two-site translation symmetry (2$\mathbb{Z}$) with $SO(3)$ intact. The ground state degeneracy (GSD) is $|\mathbb{Z}/2\mathbb{Z}| = 2$. Assuming that the length of the spin chain is even (which we do henceforth), it is known that any general vector in this two-dimensional ground space has the following non-injective MPS representation~\cite{ChenGuWen_1dClassification_PhysRevB.84.235128,Kapustin_MPSTFT_PhysRevB.96.075125}:
    \begin{multline}
        \label{eq:2site_SO(3)_noninjective}
        \ket{\psi} = \sum_{m_1\ldots m_L} Tr[\Theta~C^{m_1} D^{m_2} \ldots C^{m_{L-1}} D^{m_{L}}] \\ \ket{m_1, m_2, \ldots m_L}.
    \end{multline}
    Here, $C^m$ and $D^m$ matrices are block-diagonal as follows with the number of blocks being equal to the GSD 
    \begin{equation}
        C^m = \begin{pmatrix}
        A^m & 0 \\
        0 & B^m
        \end{pmatrix},~~~ D = \begin{pmatrix}
        B^m & 0 \\
        0 & A^m
        \end{pmatrix},
    \end{equation}
    $A^m$ and $B^m$ matrices satisfy injectivity and
    \begin{equation}
        \Theta = \begin{pmatrix}
        \alpha\mathbb{1} & 0 \\
        0 & \beta\mathbb{1}
        \end{pmatrix}\text{ with } |\alpha|^2+ |\beta|^2 = 1.
    \end{equation}
    $\Theta$ commutes with $C^{m}$ and $D^m$ and hence its location in the MPS string is irrelevant.
    That the MPS of eq~(\ref{eq:2site_SO(3)_noninjective}) is not injective can be seen by observing that
    $\Gamma_I(\Theta M) = \Gamma_I( M \Theta)$ even when $\Theta M  \neq  M \Theta$ for any $I$~\cite{Tzu-Chieh}. The ground state vector of eq~(\ref{eq:2site_SO(3)_noninjective}) can be rewritten in a more conventional form:
	\begin{figure}[!htbp]
		\centering
		\includegraphics[width=0.35\textwidth]{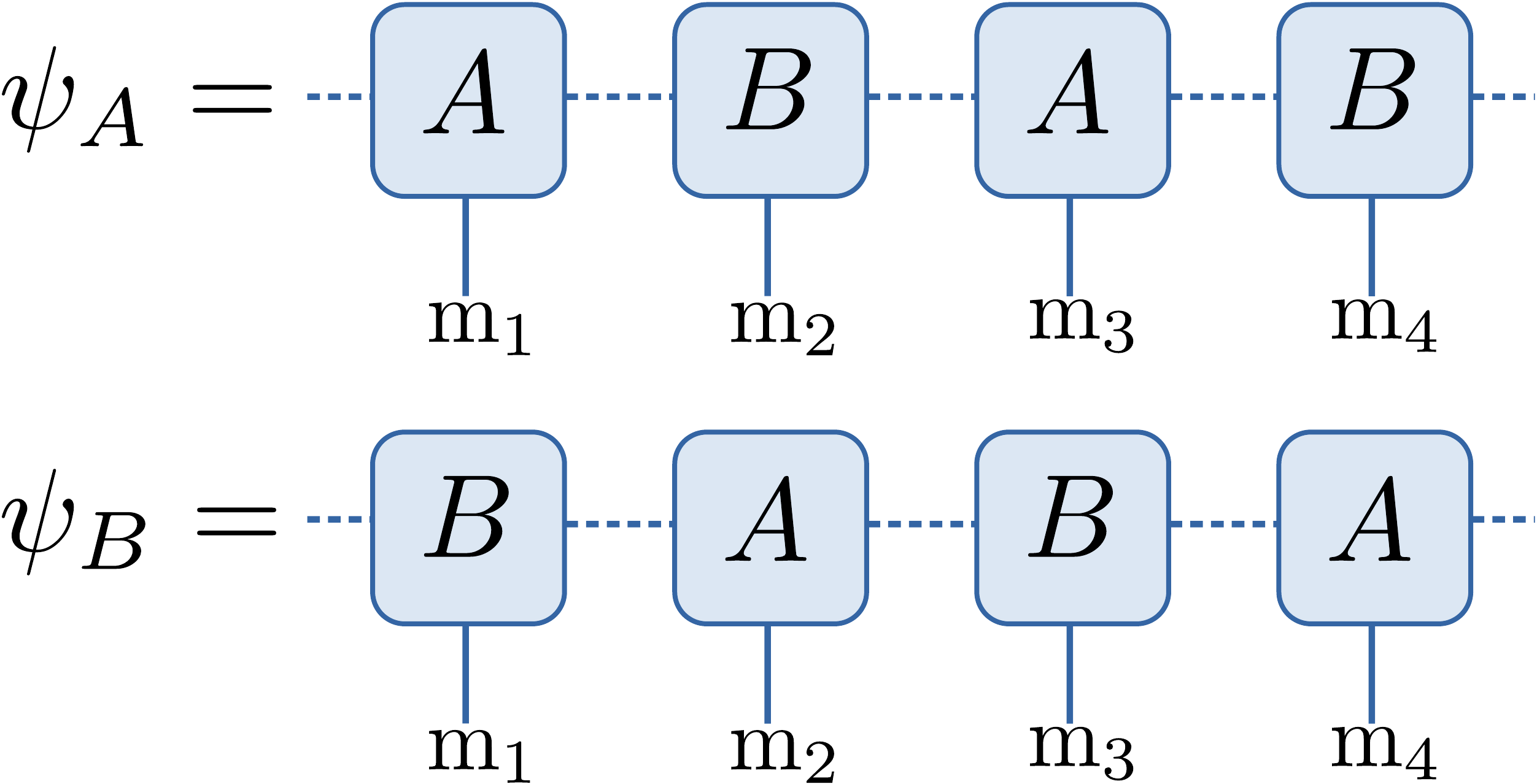}
		\caption{2-site translation invariant MPS basis states}
		\label{fig:MPS_2translation}
	\end{figure}    
    \begin{eqnarray}
    	\label{eq:2site_SO(3)_injective}
    	\ket{\psi} &=& \alpha \ket{\psi_A} + \beta \ket{\psi_B} \\
    	\ket{\psi_A} &=& \sum_{m_1\ldots m_L} Tr[ A^{m_1} B^{m_2} \ldots A^{m_{L-1}} B^{m_{L}}] \ket{m_1 \ldots m_L} \nonumber \\
    	\ket{\psi_B} &=& \sum_{m_1 \ldots m_L} Tr[ B^{m_1} A^{m_2} \ldots B^{m_{L-1}} A^{m_{L}}] \ket{m_1 \ldots m_L}. \nonumber
    \end{eqnarray}
    This tells us that the two-dimensional ground space can be spanned by two basis states, $\ket{\psi_A}$ and $\ket{\psi_B}$, orthonormal in the thermodynamic limit, each of which have an injective MPS representation.  As expected, the broken symmetry generator i.e. single site translation $\ket{m_i} \mapsto \ket{m_{i+1}}$ converts one basis state to the other $\ket{\psi_A} \mapsto \ket{\psi_B}$. We can now look at how the unbroken $SO(3)$ is implemented on the basis states. Invariance of $\ket{\psi_A}$ under the action of the $SO(3)$ transformation, $U(g) = \bigotimes_i D(g)$ implies 
    \begin{eqnarray}
    \sum_{n=-J}^{+J} D(g)_{mn} A^n = V^\dagger(g) A^m W(g)\\
    \sum_{n=-J}^{+J} D(g)_{mn} B^n = W^\dagger(g) B^m V(g).
    \end{eqnarray}
    \begin{figure}[!htbp]
    	\centering
    	\includegraphics[width=0.35\textwidth]{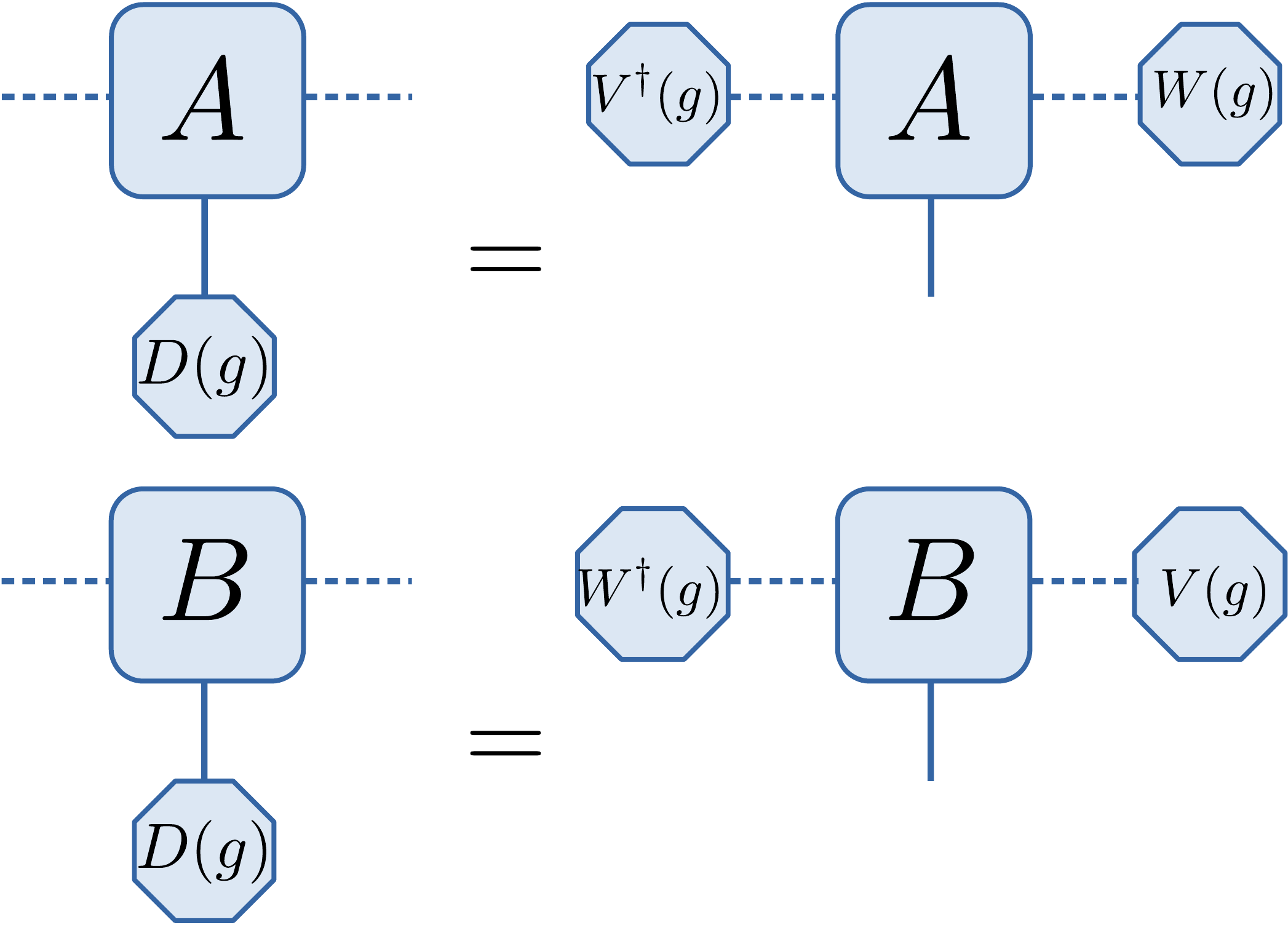}
    	\caption{$SO(3)$ transformation of 2-site translation invariant MPS basis states}
    	\label{fig:MPS_2translation_internal}
    \end{figure}
    If we now redefine the MPS grouping pairs of original spins, 
	we can write $\ket{\psi_A}$ and $\ket{\psi_B}$ in a translationally invariant form
    \begin{eqnarray}
    \ket{\psi_A} &=& \sum_{m_1 \ldots m_L} Tr[ K^{m_1 m_2} \ldots K^{m_{L-1} m_{L}}] \ket{m_1\ldots m_L}, \nonumber \\
    \ket{\psi_B} &=& \sum_{m_1 \ldots m_L} Tr[ L^{m_1 m_2} \ldots L^{m_{L-1} m_{L}}] \ket{m_1\ldots m_L} 
    \end{eqnarray}
    where, as shown in fig~(\ref{fig:MPS2site_redef})
    \begin{equation}
    	K^{mn} \equiv A^m B^n,~~~~ L^{mn} \equiv B^m A^n,
    \end{equation}
    Finally, we perform a local basis change to relabel each pair of original spins into distinct angular momentum sectors i.e.
    \begin{figure}[!htbp]
		\centering
		\includegraphics[width=0.35\textwidth]{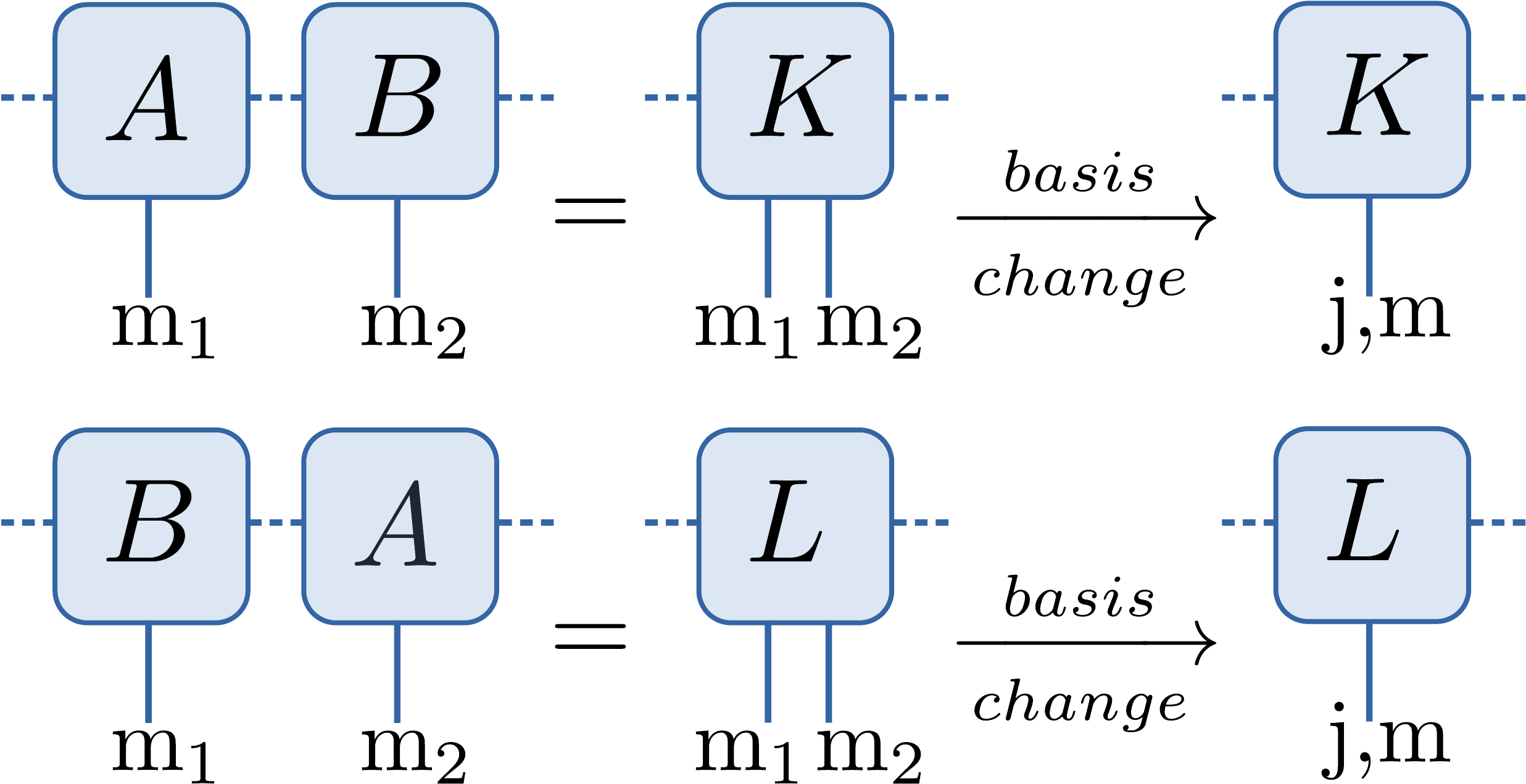}
		\caption{2-site MPS redefinition}
		\label{fig:MPS2site_redef}
	\end{figure}        
    \begin{equation}
    	\ket{m_1,m_2}= \sum_{j=0}^{2J}\sum_{m=-j}^{+j} \innerproduct{j,m}{m_1,m_2} \ket{j,m}
    \end{equation}
    using which we can write $\ket{\psi_A}$ and $\ket{\psi_B}$ as
    \begin{eqnarray}
    	\ket{\psi_A} &=& \sum_{\{j_i,m_i\}} Tr[ K^{j_1 m_1} \ldots K^{j_{\frac{L}{2}} m_{\frac{L}{2}}}] \ket{j_1 m_1\ldots j_{\frac{L}{2}} m_{\frac{L}{2}}}, \nonumber \\
    	\ket{\psi_B} &=& \sum_{\{j_i,m_i\}} Tr[ L^{j_1 m_1} \ldots L^{j_{\frac{L}{2}} m_{\frac{L}{2}}}] \ket{j_1 m_1\ldots j_{\frac{L}{2}} m_{\frac{L}{2}}} \nonumber
    \end{eqnarray}
    where
    \begin{eqnarray}
    K^{j m} &=& \sum_{m_1 m_2} \innerproduct{j,m}{m_1,m_2} K^{m_1 m_2} \\
    L^{j m} &=& \sum_{m_1 m_2} \innerproduct{j,m}{m_1,m_2} L^{m_1 m_2}. 
    \end{eqnarray}
    The length of the chain is now half the original and the local spins are $(2J+1)^2$ dimensional which we have reduced into different $j$ sectors. $j$ is now \emph{only} integer valued and varies from $0,1,\ldots 2J$. The change of basis block-diagonalizes the two-site $SO(3)$ transformation matrices as follows
    \begin{equation}
    D(g) \otimes D(g) \xrightarrow{\text{basis change}} \bigoplus_{j = 0}^{2 J} D^j(g)
    \end{equation}
    We can now impose the invariance of the injective, translation invariant basis states under $SO(3)$ transformations:
   	\begin{figure}[!htbp]
	\centering
	\includegraphics[width=0.35\textwidth]{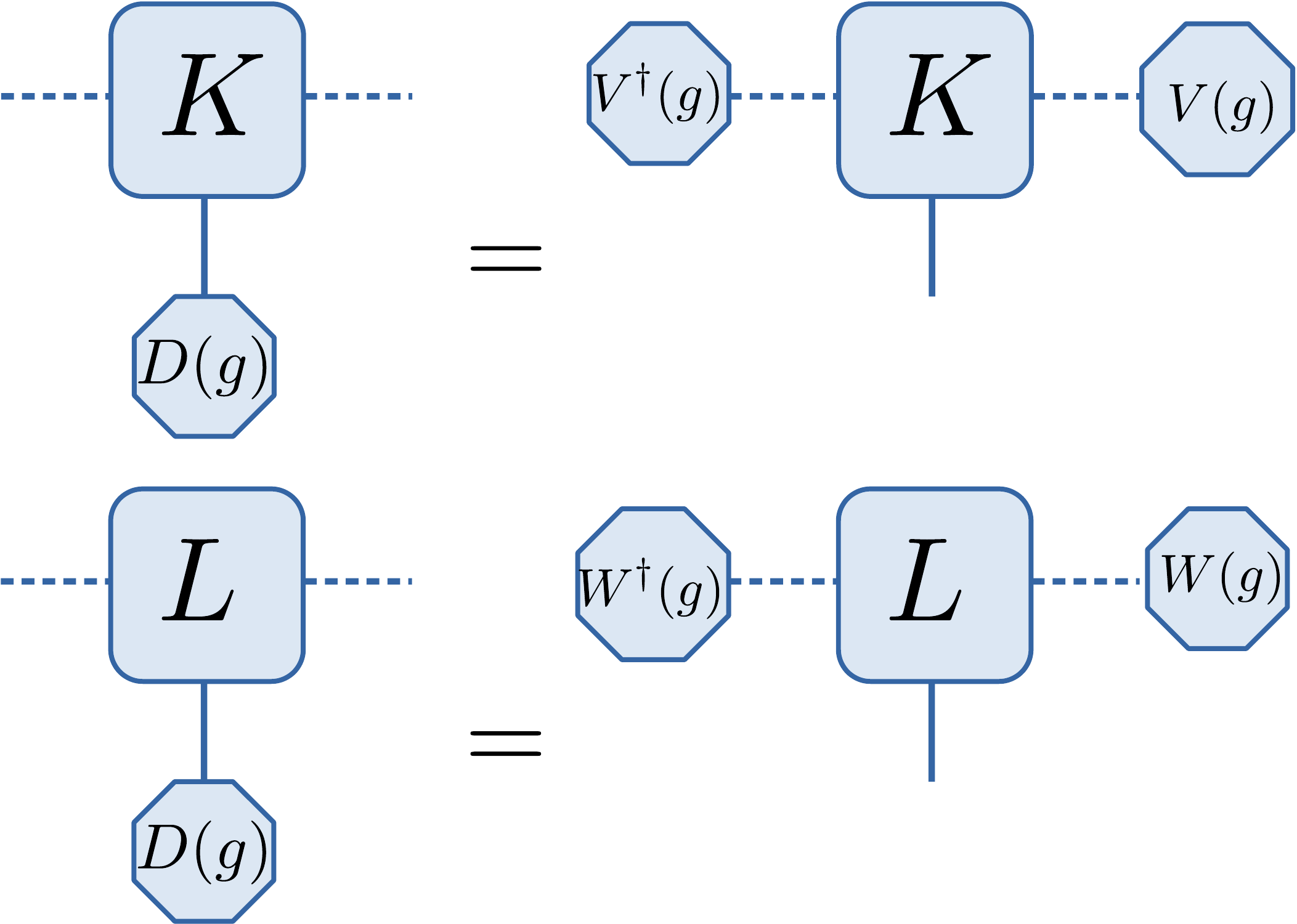}
	\caption{$SO(3)$ transformation on redefined MPS}
	\label{fig:MPS2site_redef_internal}
	\end{figure}    
    \begin{eqnarray}
    \sum_{j = 0}^{2 J}\sum_{n= -j}^{+j} D^j(g)_{m  n} K^{j n} = V^\dagger(g) K^{j m}V(g),\\
    \sum_{j = 0}^{2 J} \sum_{n= -j}^{+j} D^j(g)_{m  n} L^{j n} = W^\dagger(g) L^{j m}W(g).
    \end{eqnarray}
    Once again, invoking the Wigner-Eckart theorem, we can write the matrix elements of $K$ and $L$ matrices as
	\begin{multline}
		\innerproduct{j_\alpha,m_\alpha,d_\alpha|K^{jm}}{j_\beta,m_\beta,d_\beta} \\= \innerproduct{j,m;j_\beta,m_\beta}{j_\alpha,m_\alpha} \innerproduct{j_\alpha,d_\alpha||K^{jm}|}{j_\beta,d_\beta}
	\end{multline}
	\begin{multline}
		\innerproduct{j_\alpha,m_\alpha,d_\alpha|L^{jm}}{j_\beta,m_\beta,d_\beta} \\= \innerproduct{j,m;j_\beta,m_\beta}{j_\alpha,m_\alpha} \innerproduct{j_\alpha,d_\alpha||L^{jm}|}{j_\beta,d_\beta}
	\end{multline}
	Now, since $j$ is integer, the CG coefficients $\innerproduct{j,m;j_\beta,m_\beta}{j_\alpha,m_\alpha}$  exist and there is no problem. Thus, for half-odd-integer spin chains, a two-fold-degenerate gapped ground state is possible retaining $SO(3)$ invariance but breaking translation invariance to two-site translation invariance. 
	
	This result can easily be generalized to prove that if single site translation invariance $\bZ$ is broken to even-site translation invariance $2n \bZ$, a $2n$ fold $SO(3)$ invariant ground space is possible following the same arguments as above because fusing an even number of half-odd-integer spins can be CG decomposed to give only integer spins. However, if translation invariance  $\bZ$ is broken to odd-site translation invariance $(2n+1) \bZ$, a $2n+1$-fold degenerate ground state with a gap is not possible because fusing an odd number of half-odd-integer spin, upon CG decomposition results in only half-odd-integer spins and the arguments of theorem~(\ref{th:SO(3)LSM}) would again lead to the vanishing of the MPS. 
\end{proof}

\section{Proof of the general case}
\label{sec:G_LSM}
We now state and prove a generalization of theorem~(\ref{th:SO(3)LSM}) where the group $SO(3)$ is generalized to any group $G$, integer representations are generalized to so-called \emph{linear} representations of $G$ and half-integer representations are generalized to so-called \emph{projective} representations of $G$. We briefly review projective and linear representations now and direct the reader to appendix~(\ref{app:projective representations}) for more details. 

A linear representation of $G$ is an assignment of an invertible complex matrix, $M(g)$ for each element $g \in G$ such that they are compatible with group multiplication as follows
\begin{equation}
    \label{eq:linear_representation}
    M(g)M(h)= M(gh) 
\end{equation}
A projective representation is an assignment of an invertible complex matrix, $V(g)$ for each element $g \in G$ such that group multiplication holds only \emph{upto a complex phase}
\begin{equation}
    \label{eq:projective_representation}
    V(g) V(h) = \omega(g,h) V(gh)
\end{equation}
Projective representations are familiar in point-particle quantum mechanics because, given the Hilbert space of the system, physical states correspond to the rays in the Hilbert space and thus, even if the vectors transform as a projective representation, all observables transform as linear representations. Let us state a few important facts about projective representations.
\begin{enumerate}
    \item In general, the complex phases $\omega(g,h)$ cannot be removed by rephasing $V(g) \mapsto e^{i \theta_g} V(g)$. 
    \item Different \emph{classes} of projective representations can be assigned elements of an Abelian group, $\hgc$. Linear representations of $G$ correspond to the identity element, $e \in \hgc$.  
    \item The group multiplication law of $\hgc$ is reflected in the \emph{fusion rules} of the projective representations of $G$.
    \item The linear and projective representations of $G$ can be lifted to the linear representation of a \emph{covering group}, $\tilde{G}$.  
 \end{enumerate}
As an example, $G =SO(3)$ has one class of linear representations (integer spins) and one class of projective representations (half-odd-integer spins) both of which can be lifted to linear representations of $\tilde{G} = SU(2)$.  The group structure of $\hgcinput{SO(3)} = \bZ_{2}$ and is reflected in the fact that fusion of integer and half-odd-integer irreps of $SU(2)$ has a $\bZ_{2}$ structure. Let us see this in detail. Let us represent the class of half-odd-integers irreps as $\left[\frac{1}{2}\right]$ and the class of integer irreps as $[1]$. The claim is that each of these classes can be assigned one element of $\pm 1 \in \ztwo $ as follows
\begin{equation}
[1]\rightarrow +1,~~~~\left[\frac{1}{2} \right] \rightarrow -1
\end{equation}
and the $\ztwo$ group structure is reflected in the fusion rules as follows
\begin{equation}
\begin{Bmatrix}
[1] \otimes [1] \in [1] \\[0.2em]
[1] \otimes \left[\frac{1}{2} \right] \in \left[\frac{1}{2} \right] \\[0.2em]
\left[\frac{1}{2} \right] \otimes [1] \in \left[\frac{1}{2} \right] \\[0.2em]
\left[\frac{1}{2} \right] \otimes \left[\frac{1}{2} \right] \in [1]
\end{Bmatrix} \rightarrow
\begin{Bmatrix}
+1 \times +1 = +1 \\[0.2em]
+1 \times -1 = -1 \\[0.2em]
-1 \times +1 = -1 \\[0.2em]
-1 \times -1 = +1 
\end{Bmatrix}
\end{equation}
We now state and prove the general theorem. We retain notations used in the proof of theorem~(\ref{th:SO(3)LSM}) for convenience. 

\begin{theorem}
\label{th:General_LSM}
A spin chain with lattice translation invariance and a projectively realized on-site $G$ symmetry cannot have a unique gapped ground state.
\end{theorem}
\begin{proof}
We prove this by contradiction following the steps used in proving theorem~(\ref{th:SO(3)LSM}). Let us assume that there does exist a unique gapped ground state with an MPS representation of the form (\ref{eq:MPS_onsite_translation_invariance}) invariant under on-site and translation symmetries:
\begin{equation}
    \label{eq:G_MPS_onsite_translation_invariance}
    \sum_{n=1}^{|J|} D(g)_{mn} A^{n} = V^\dagger(g) A^{m} V(g)
\end{equation}
and satisfying the injectivity condition of eq~(\ref{eq:injectivity}).
We assume that the physical spin transforms as an irrep, which we label $J$ that belongs to some non-trivial class of projective representations i.e.
\begin{equation}
\label{eq:projective_assumption}
[J] = \omega_J \in \hgc,~~~\omega_J \neq e.
\end{equation}
$m$ labels its internal states $m=1,2,\ldots,|J|$ and $D(g)$ is the $|J| \times |J|$ matrix representation of the group elements~\cite{irrep_note}. We can choose an appropriate basis such that $V(g)$ are block diagonalized into different irrep sectors to examine the matrix elements of $A^m$. Just like in the case of $SO(3)$, the matrix elements are constrained by an equivalent of the Wigner-Eckart theorem~\cite{AP_TzuChieh_PhysRevA.92.022310,SinghPfeiferVidalPhysRevA.82.050301}:
\begin{multline}
    \innerproduct{j_\alpha,m_\alpha,d_\alpha|A^m}{j_\beta,m_\beta,d_\beta} \\= \innerproduct{J,m;j_\beta,m_\beta}{j_\alpha,m_\alpha} \innerproduct{j_\alpha,d_\alpha||A^m|}{j_\beta,d_\beta}.
\end{multline}
$\innerproduct{J,m;j_\beta,m_\beta}{j_\alpha,m_\alpha}$ are the Clebsch-Gordan coefficients corresponding to the fusion of irreps $J \otimes j_\beta \rightarrow j_\alpha$ and $\innerproduct{j_\alpha,d_\alpha||A^m|}{j_\beta,d_\beta}$ are undetermined from symmetry but are independent of $m$. The bond Hilbert space can contain multiple copies of same irrep which is labeled by $d_\alpha$ and $d_\beta$. 

Injectivity constrains the bond irreps $j_\alpha, j_\beta$ to belong to a \emph{definite} class of projective representations i.e. be labeled by  a specific element of $\omega \in \hgc$ such that $[j_\alpha] = [j_\beta] = \omega$. This is proven in  lemma~(\ref{lemma:G_Bond_hilbert_space}) of appendix(\ref{app:proofs}). For the CG coefficients $\innerproduct{j_\alpha,d_\alpha||A^m|}{j_\beta,d_\beta}$ to exist, the corresponding fusion $J \otimes j_\beta \rightarrow j_\alpha$ should be consistent with the multiplication of the group elements of $\hgc$ that labels the irreps i.e.
\begin{equation}
    \label{eq:CG rule for projective physical spin}
    [J] \otimes [j_\beta] \in [j_\alpha] \implies \omega_J \times \omega = \omega 
\end{equation}
 Clearly, eq~(\ref{eq:CG rule for projective physical spin}) is only consistent with $\omega_J = e$ i.e. the physical spin is a linear representation of $G$. This contradicts our original assumption in eq~(\ref{eq:projective_assumption}) and thus the CG coefficients vanish along with the matrices $A^m$ and the state $\ket{\psi}$ itself. This means that our assumption of the existence of a unique, symmetric, gapped ground state itself was wrong.
\end{proof} 
Now that we have established that the setting of theorem~(\ref{th:General_LSM}) forbids a unique gapped ground state, we again enumerate the possibilities for the low energy properties of the spin chain. 

\begin{corollary}
	A spin chain with lattice translation invariance and a projectively realized on-site G symmetry can be either gapless or have a degenerate ground space with spontaneously broken symmetries.
\end{corollary}
\begin{proof}
	 Theorem~(\ref{th:General_LSM}) says that lattice translations and  projective on-site symmetries are not compatible with the existence of an injective MPS. One possibility is that the ground state simply does not have an MPS description and the system is gapless. Another is that the system has a gapped spectrum but the ground state is not unique and symmetry is spontaneously broken. This admits an MPS description for the vectors of the ground space but they are generally not injective. For a general on-site $G$, unlike the case of $SO(3)$, in addition to the possibility that lattice translations are spontaneously broken, a finite on-site symmetry could also be spontaneously broken~\cite{SSB_note,AP_MBL_PhysRevB.96.165136}. 
	 
    Let us first consider the conditions under which a gapped ground state is possible when lattice translation invariance is spontaneously broken. In the $SO(3)$ case, it was argued that if single-site lattice translation invariance is broken down to two-site lattice translation invariance, we obtain a two-dimensional ground space with a spectral gap. A choice of basis states was obtained such that, by defining \emph{superspins} consisting of two originally half-odd-integer spins, they could be made translation invariant with on-site integer representation and $SO(3)$ invariance. We now generalize this to general groups. 
    
    Like we did in theorem~(\ref{th:General_LSM}), let us assume that the physical spin transforms as a non-trivial projective irrep $J$ that corresponds to a class $[J]$ that can be assigned a non-trivial group element $\omega_J \in \hgc$, $\omega_J \neq e$.  If single-site translation invariance $\bZ$ is spontaneously broken down to f-site translation invariance $f\bZ$, the GSD is $|\bZ/f \bZ| = f$. A vector in the f-dimensional ground space can be represented as~\cite{ChenGuWen_1dClassification_PhysRevB.84.235128,Kapustin_MPSTFT_PhysRevB.96.075125}
    \begin{multline}
        \label{eq:fsite_G_noninjective}
        \ket{\psi} = \sum_{m_1,m_2,\ldots,m_L} Tr[\Theta C_1^{m_1} C_2^{m_2} \ldots C_f^{m_f} C_1^{m_{f+1}} \ldots] \\ \ket{m_1, m_2, \ldots m_L}.
    \end{multline}
    Here, $C_k^m$ matrices ($k = 1 \ldots f$) are block-diagonal with the number of blocks being equal to the GSD i.e. $f$ as follows
    \begin{eqnarray}
        C_1^m &=& \begin{pmatrix}
        A_1^m & 0 & \ldots &0 \\
        0 & A_2^m & \ldots &0 \\
        \vdots & \vdots & \ddots  &0 \\
        0 &   \ldots & A_{f-1}^m & 0 \\        
        0 &  \ldots & 0& A_f^m
        \end{pmatrix}, \nonumber\\
        C_2^m &=& \begin{pmatrix}
        A_2^m & 0 & \ldots &0 \\
        0 & A_3^m & \ldots &0 \\
        \vdots & \vdots & \ddots  &0 \\
        0 &   \ldots & A_f^m & 0 \\
		0 &   \ldots & 0& A_1^m        
        \end{pmatrix}, \nonumber \\
         C_{k|f}^m &=& \begin{pmatrix}
        A_{k|f}^m & 0 & \ldots &0 \\
        0 & A_{k+1|f}^m & \ldots &0 \\
        \vdots & \vdots &\ddots  &0 \\
    0 &   \ldots & A_{k-2|f}^m & 0 \\        
        0 &   \ldots & 0& A_{k-1|f}^m        
        \end{pmatrix} 
    \end{eqnarray}
    where, $A_k^m$ are injective, $k|f$ is shorthand for $k$ mod $f$ and
    \begin{equation}
        \Theta = \begin{pmatrix}
        \alpha_1  \mathbb{1} & 0 & \ldots &0 \\
        0 & \alpha_2  \mathbb{1} & \ldots &0 \\
        \vdots & \vdots &\ddots  &0 \\
        0 &   \ldots & 0& \alpha_k  \mathbb{1}        
        \end{pmatrix} \text{ with } \sum_{k=1}^f|\alpha_k|^2 = 1
    \end{equation}
    $\Theta$ commutes with the matrices $C_k^{m_i}$ and hence its location in the MPS string is irrelevant. That the MPS of eq~(\ref{eq:fsite_G_noninjective}) is not injective can be seen by observing that
    $\Gamma_I(\Theta M) = \Gamma_I( M \Theta)$ even when $\Theta M  \neq  M \Theta$ for any $I$~\cite{Tzu-Chieh}. The ground state vector of eq~(\ref{eq:fsite_G_noninjective}) can be rewritten as
    \begin{figure}[!htbp]
	\centering
	\includegraphics[width=0.35\textwidth]{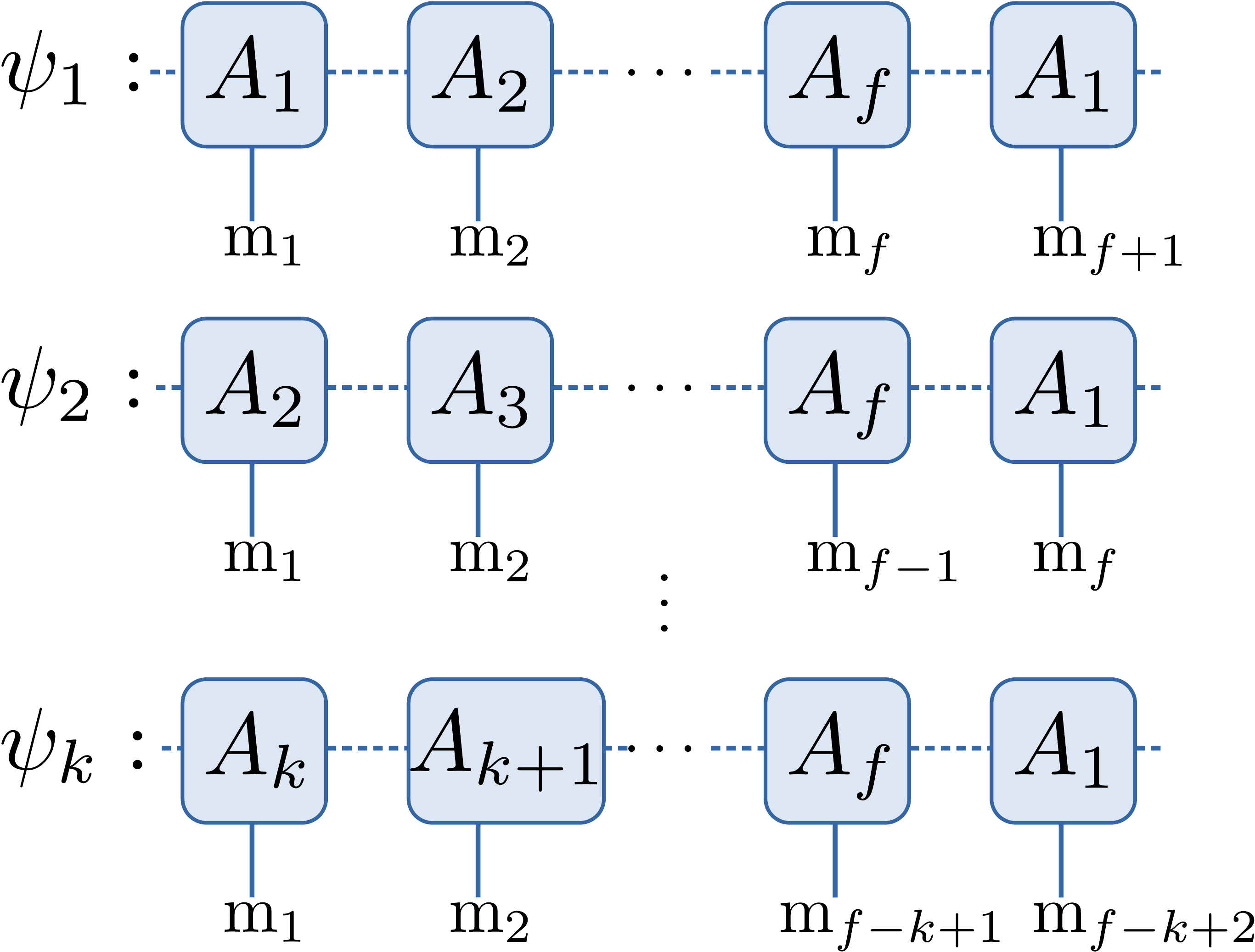}
	\caption{f-site translation invariant MPS basis state}
	\label{fig:MPS_fftranslation}
	\end{figure}    
    \begin{eqnarray}
    \label{eq:fsite_G_injective}
    \ket{\psi} &=& \sum_{k=1}^f \alpha_k \ket{\psi_k}  \\
    \ket{\psi_k} &=& \sum_{m_1\ldots m_L} Tr[ A_{k|f}^{m_1} A_{k+1|f}^{m_2} \ldots] \ket{m_1 \ldots m_L} 
    \end{eqnarray}
    This tells us that the f-dimensional ground space can be spanned by $f$ basis states, $\ket{\psi_k}$, orthonormal in the thermodynamic limit, which have an injective MPS representation.  As expected, the broken symmetry generator i.e. single site translation $\ket{m_i} \mapsto \ket{ m_{i+1}}$ converts one basis state to the other $\ket{\psi_k} \mapsto \ket{\psi_{k-1|f}}$. We can now look at how the unbroken internal $G$ symmetry is implemented on the basis states. Invariance of $\ket{\psi_k}$ under the action of the transformation, $U(g) = \bigotimes_i D(g)$ implies 
    \begin{figure}[!htbp]
    	\centering
    	\includegraphics[width=0.35\textwidth]{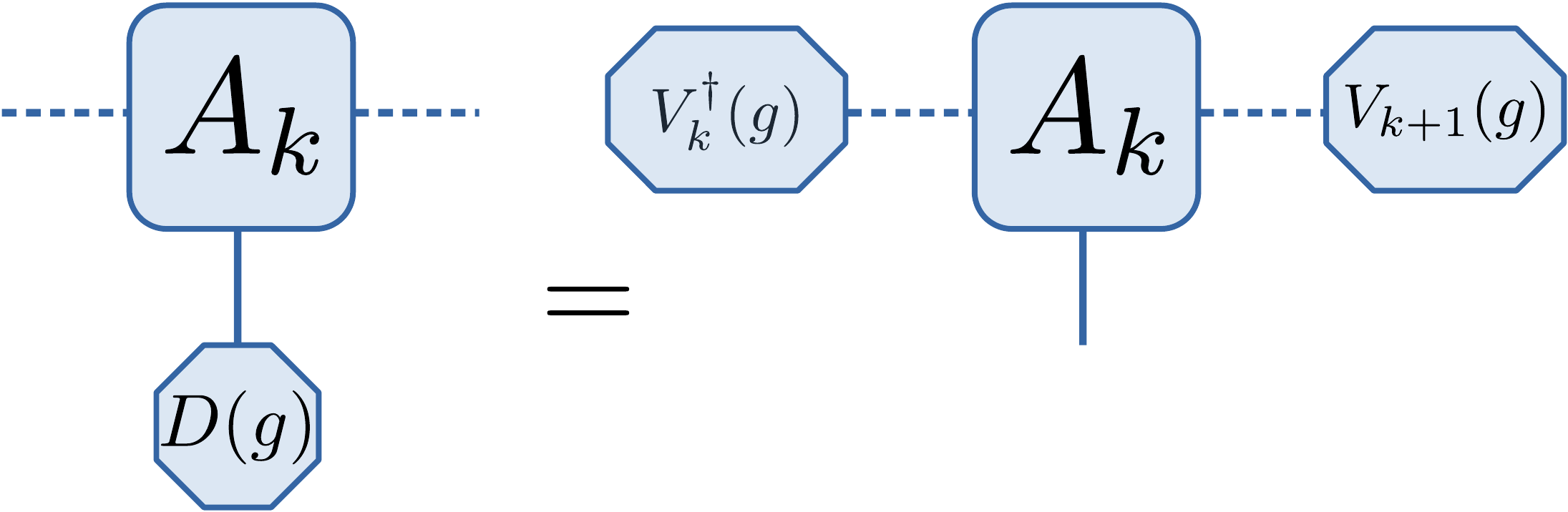}
    	\caption{$G$ transformation of f-site translation invariant MPS basis states}
    	\label{fig:MPS_ftranslation_internal}
    \end{figure}
    \begin{eqnarray}
    \sum_{n=1}^{|J|} D(g)_{mn} A_k^n = V_{k}^\dagger(g) A_k^m V_{k+1}(g)
    \end{eqnarray}
    If we now redefine the MPS grouping $k$ spins, we can write $\ket{\psi_k}$ in a translationally invariant form
    \begin{eqnarray}
    \ket{\psi_k} &=& \sum_{m_1 \ldots m_L} Tr[ L_k^{m_1 \ldots m_f} \ldots L_{k}^{m_{L-f} \ldots m_{L}}] \ket{m_1\ldots m_L}, \nonumber \\
    \end{eqnarray}
	\begin{figure}[!htbp]
		\centering
		\includegraphics[width=0.35\textwidth]{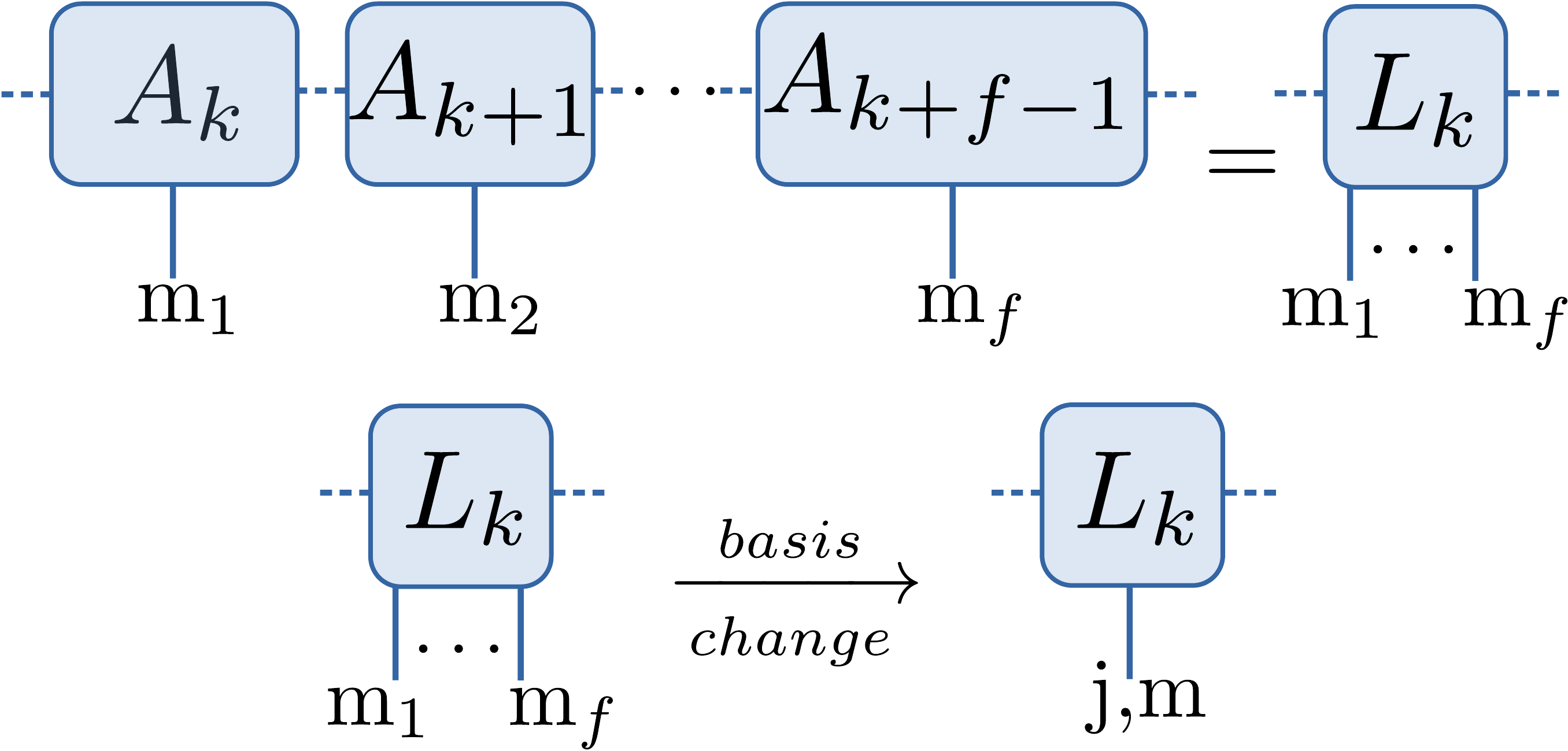}
		\caption{f-site MPS redefinition}
		\label{fig:MPSfsite_redef}
	\end{figure}    
    where, as shown in fig~(\ref{fig:MPSfsite_redef})
    \begin{equation}
    L_k^{m_1 \ldots m_f} \equiv A_k^{m_1} A_{k+1}^{m_2} \ldots A_{k+f-1}^{m_f} \nonumber
    \end{equation}
    Finally, we perform a local basis change to relabel each superspin into distinct irreps
    \begin{equation}
    \ket{m_1,m_2,\ldots,m_f} \xrightarrow{\text{basis~change}} \ket{j,m}
    \end{equation}
    using which we can write $\ket{\psi_k}$ as
    \begin{eqnarray}
    \ket{\psi_k} &=& \sum_{\{j_i,m_i\}} Tr[ L^{j_1 m_1} \ldots L^{j_{\frac{L}{f}} m_{\frac{L}{f}}}] \ket{j_1 m_1\ldots j_{\frac{L}{f}} m_{\frac{L}{f}}}\nonumber
    \end{eqnarray}
    The length of the chain is now $1/f$ the original (we assume that the length of the spin chain is a multiple of $f$) and the local spins are $|J|^f$ dimensional which we have reduced into different $j$ sectors. The change of basis block-diagonalizes the f-site $G$ transformation matrices as follows
    \begin{equation}
    \bigotimes_{l=1}^f D(g)  \xrightarrow{\text{basis change}} \bigoplus_{j \in  J^{\otimes f }} D^j(g)
    \end{equation}
    We can now impose the invariance of the injective, translation invariant basis states under $G$ transformations.
	\begin{figure}[!htbp]
		\centering
		\includegraphics[width=0.35\textwidth]{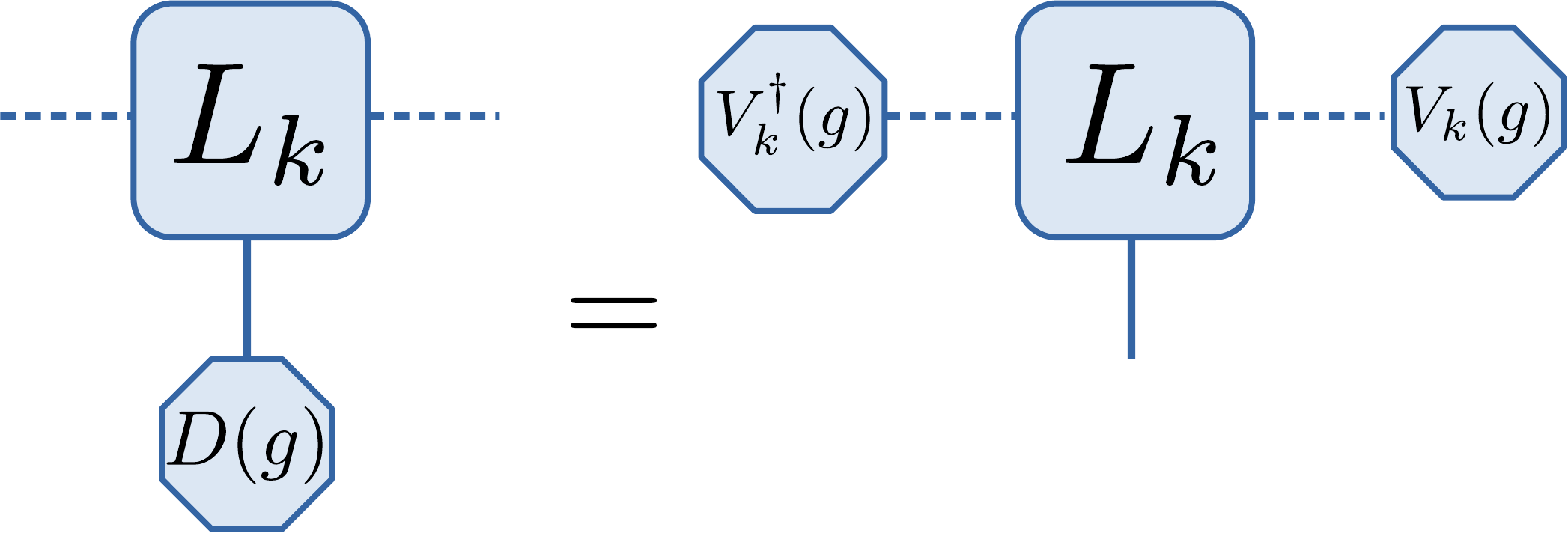}
		\caption{$G$ transformation on redefined MPS}
		\label{fig:MPSfsite_redef_internal}
	\end{figure}    
    \begin{eqnarray}
    \sum_{j \in  J^{\otimes f }} \sum_{n =1}^{|j|} D^j(g)_{m  n} L_k^{j n} = V_k^\dagger(g) L_k^{j m}V_k(g)
    \end{eqnarray}
    Once again, invoking the Wigner-Eckart theorem, we can write the matrix elements of $L^{jm}_k$ matrices as
    \begin{multline}
    \innerproduct{j_\alpha,m_\alpha,d_\alpha|L_k^{jm}}{j_\beta,m_\beta,d_\beta} \\= \innerproduct{j,m;j_\beta,m_\beta}{j_\alpha,m_\alpha} \innerproduct{j_\alpha,d_\alpha||L_k^{jm}|}{j_\beta,d_\beta}
    \end{multline}
    For the CG coefficients $\innerproduct{j,m;j_\beta,m_\beta}{j_\alpha,m_\alpha}$ to exist, following the logic of the proof of theorem~(\ref{th:General_LSM}), $[j]$ should correspond to the identity element of $\hgc$. In other words, $j$ should be a linear irrep of $G$. Thus, for the CG coefficients to exist, $f$ should be such that $J^{\otimes f}$ has only linear irreps. Determining this value of $f$ is easy if the group $\hgc$ and its element corresponding to the class of projective irreps $[J]= \omega_J$ is known. Since $\hgc$ is a finite Abelian group, $\omega_J$ has some \emph{finite order} i.e. a finite integer $F$ such that $\omega^F =e$. Thus, if and only if $f$ is a multiple of $F$, $j$  corresponds to linear irreps and the CG coefficients exist. Any other $f$ would produce $j$ that corresponds to a non-trivial projective irrep and the arguments of theorem~(\ref{th:General_LSM}) would again lead to the vanishing of the MPS.
    
    To summarize, the spin chain can have a gapped spectrum if single-site lattice translation symmetry $\bZ$ is broken down to $nF\bZ$ producing a $nF$-fold GSD where $F$ is the order of the group element of $\hgc$ that characterizes the projective representation of the physical spin and $n \in \bZ$. 
    
    We now consider the possibility that on-site symmetry $G$ is spontaneously broken to one of its subgroups, $H \subset G$ with translation invariance intact. In this case, the GSD is $|G/H|$ and any vector in the  ground space can be given the following non-injective MPS representation
    \begin{equation}
    \label{eq:H_noninjective}
    \ket{\psi} = \sum_{m_1 \ldots m_L} Tr[\Theta C^{m_1}C^{m_2}\ldots C^{m_L}] \ket{m_1 \ldots m_L}
    \end{equation} 
    where $C^m$ is block diagonal with $|G/H|$ blocks
    \begin{equation}
    C^m = \begin{pmatrix}
    A_1^m & 0 & \ldots &0 \\
    0 & A_2^m & \ldots &0 \\
    \vdots & \vdots & \ddots  &0 \\
    0 &  \ldots & 0& A_{|G/H|}^m
    \end{pmatrix},
    \end{equation}
	$A^m_k$ are injective and   
	\begin{equation}
	\Theta = \begin{pmatrix}
	\alpha_1  \mathbb{1} & 0 & \ldots &0 \\
	0 & \alpha_2  \mathbb{1} & \ldots &0 \\
	\vdots & \vdots &\ddots  &0 \\
	0 &   \ldots & 0& \alpha_{|G/H|}  \mathbb{1}        
	\end{pmatrix} \text{ with } \sum_{k=1}^{|G/H|}|\alpha_k|^2 = 1
	\end{equation}
	Again, $\Theta$ commutes with the matrices $C^{m_i}$ and hence its location in the MPS string is irrelevant. That the above MPS is not injective can be seen by observing that $\Gamma_I(\Theta M) = \Gamma_I( M \Theta)$ even when $\Theta M  \neq  M \Theta$ for any $I$~\cite{Tzu-Chieh}. The ground state vector of eq~(\ref{eq:H_noninjective}) can be rewritten as
	\begin{eqnarray}
	\label{eq:H_injective}
	\ket{\psi} &=& \sum_{k=1}^{|G/H|} \alpha_k \ket{\psi_k}  \\
	\ket{\psi_k} &=& \sum_{m_1\ldots m_L} Tr[ A_{k}^{m_1} A_{k}^{m_2} \ldots A_{k}^{m_L}] \ket{m_1 \ldots m_L} .
	\end{eqnarray}
	This tells us that the $|G/H|$ dimensional ground space can be spanned by $|G/H|$ basis states $\ket{\psi_k}$,  orthonormal in the thermodynamic limit, which have an injective, translation invariant MPS representation. The elements of the coset $\rho \in G/H$ faithfully permute the basis states
	\begin{equation}
	\ket{\psi_k} \xrightarrow{\rho} \ket{\psi_{\rho(k)}}.
	\end{equation}
	On the other hand, $h \in H$ is still a symmetry. This means that there exist bond representations, $V(h)$ such that the following holds
	\begin{equation}
		\label{eq:H_MPS_onsite_translation_invariance}
		\sum_{n=1}^{|J|} D(h)_{mn} A_k^{n} = V^\dagger(h) A_k^{m} V(h)
	\end{equation}
	However, at this stage, we must note that while the original physical spin $J$ was a projective \emph{irrep} of $G$, it may be a reducible representation of  $H$. Thus, we choose a basis such that when we restrict the group elements to the subgroup $H$, the matrix elements of $D(h)$ are block diagonalized in $H$ irreps. 
	\begin{equation}
	D(h) = \bigoplus_{j \in J} D^j(h)
	\end{equation}
	and eq~(\ref{eq:H_MPS_onsite_translation_invariance}) reduces to 
	\begin{equation}
	\label{eq:H_MPS_onsite_translation_invariance}
	\sum_{j \in J}\sum_{n=1}^{|j|} D^j(h)_{mn} A_k^{j n} = V^\dagger(h) A^{j m} V(h)
	\end{equation}
	Once again, we choose a basis so that $V(h)$ too are block-diagonalized into $H$ irreps all of which are constrained by injectivity to belong to the same projective class in lemma~(\ref{lemma:G_Bond_hilbert_space}) and invoke the Wigner-Eckart theorem to write down the matrix elements of $A^{jm}_k$ as
	\begin{multline}
		\innerproduct{j_\alpha,m_\alpha,d_\alpha|A_k^{jm}}{j_\beta,m_\beta,d_\beta} \\= \innerproduct{j,m;j_\beta,m_\beta}{j_\alpha,m_\alpha} \innerproduct{j_\alpha,d_\alpha||A^{jm}_k|}{j_\beta,d_\beta}
	\end{multline}

Now, it is easy to determine which $H \in G$ allows the above CG coefficients. Even though the original physical spins transformed as non-trivial projective irreps of $G$, when restricted to the group elements of $H$, the $H$ irreps, $j \in J$ may or may not be projective as determined by the element of $\hgcinput{H}$ that labels their class $[j]$. If this element, $\omega_j$ is trivial, and the $H$ irreps are all linear, the CG coefficients do not vanish and a gapped spectrum with $|G/H|$ degenerate ground space with MPS representation is allowed. Such a subgroup $H$ \emph{always} exists (The trivial group $\{e\}$ being a trivial example. See appendix~(\ref{app:projective representations}). On the other hand, if $\omega_j$ is a non-trivial element and $j$ is a projective representation of $H$, the arguments of theorem~(\ref{th:General_LSM}) lead to the vanishing of the MPS. It is easy to generalize the above considerations to the case when there is a combination of on-site and translation symmetry breaking leading to more interesting possibilities. 

\end{proof}

To summarize, we have found the following possibilities for the spectrum of the spin chain with projective on-site and lattice translation symmetries:
\begin{enumerate}
	\item Gapless
	\item Gapped with translation symmetry $\bZ$ spontaneously broken to $f \bZ$ with the on-site symmetry $G$ unbroken, where $f$ is an integer multiple of the order of the group element of $\hgc$ that characterizes the projective irrep of the physical spin.
	\item Gapped with on-site symmetry $G$ spontaneously broken to $H$ and translation invariance $\bZ$ unbroken such that when restricted to $H$, the representation of the group is linear. 
	\item A combination of the above two. 
\end{enumerate}

\section{LSM, DQC and SPT}
\label{sec:DQC}
In his original approach to prove his famous `conjecture'~\cite{Haldane2016ground} about the nature of the spectrum for integer and half-odd-integer QAFs, Haldane actually considered the `XXZ' deformation of the QAF of eq~(\ref{eq:QAF}):
\begin{equation}
\label{eq:XXZ}
H =J \sum_i \left(S^x_{i} S^x_{i+1} + S^y_{i} S^y_{i+1} + \Delta S^z_{i} S^z_{i+1} \right)
\end{equation}
which matches eq~(\ref{eq:QAF}) when $\Delta =1$. For $\Delta \neq 1$, spin rotation symmetry is explicitly broken. It was found that for $\Delta >0 $ there existed two phases- a gapless phase ($0<\Delta <1$) and a N\'{e}el phase ($\Delta >1$) separated by QAF point $\Delta = 1$. As suggested by the lack of a trivial phase, it turns out that LSM constraints operate here too and it is interesting to analyze the protecting on-site symmetry that forbids a trivial phase. The residual $SO(2)$ symmetry corresponding to spin rotations of the form 
\begin{equation}
R(\theta):\begin{pmatrix}
S^x \\
S^y \\
S^z
\end{pmatrix} \mapsto 
\begin{pmatrix}
\cos \theta & -\sin \theta & 0 \\
\sin \theta & ~~\cos \theta & 0 \\
0 & 0 & 1
\end{pmatrix}
\begin{pmatrix}
S^x \\
S^y \\
S^z
\end{pmatrix}
\end{equation}
is an important analytical tool at the heart of bosonization techniques~\cite{Haldane2016ground,Oshikawa_LSM_PhysRevLett.84.1535} and defining `twist operators'~\cite{tasaki2018lieb} which have helped prove LSM theorems. However, it is easy to see that the mere presence of $SO(2)$ alone does not forbid a trivial phase- we can add a term $\lambda \sum_i S^z_i$ to the Hamiltonian~(\ref{eq:XXZ}) and tune $\lambda >> \{J,\Delta\}$ and get a trivial phase with product state ground state. However, (\ref{eq:XXZ}) has an additional symmetry 
\begin{equation}
\label{eq:Px}
P_x:\begin{pmatrix}
S^x \\
S^y \\
S^z
\end{pmatrix} \mapsto \begin{pmatrix}
~~S^x \\
-S^y \\
-S^z
\end{pmatrix}
\end{equation}
which does not commute with the SO(2) rotations as $P_x R(\theta) P_x =R(-\theta)$ and forms the group $O(2)$ which is projectively realized. It is this symmetry that forbids terms like $\lambda \sum_i S^z_i$ and the existence of the trivial phase. In a different work, Haldane~\cite{Haldane_PhysRevB.25.4925_NeelDimer} added a next-nearest neighbour Heisenberg coupling to the XXZ model of eq~(\ref{eq:XXZ}) as follows
\begin{equation}
H = \sum_i J_1 \left(S^x_{i} S^x_{i+1} + S^y_{i} S^y_{i+1} + \Delta S^z_{i} S^z_{i+1} \right) + J_2 \vec{S}_i.\vec{S}_{i+1} \nonumber
\end{equation}
and discovered that there existed a continuous phase transition between the N\'{e}el and valence-bond-solid (VBS) phases. These phases break incompatible symmetries (meaning the residual symmetry of one phase is not the subgroup of the other) and are hence \emph{Landau-forbidden}. This is a one dimensional version of the famous \emph{deconfined quantum critical} (DQC)~\cite{Senthil_DQC1490} N\'{e}el-to-VBS transition in two dimensions.

In fact, we could break all continuous symmetries and only retain the $\ztzt$ symmetry generated by $P_x$, $P_y$ and $P_z$ for LSM to hold~\cite{TasakiOgata_LSM_2019} where 
\begin{equation}
P_y:\begin{pmatrix}
S^x \\
S^y \\
S^z
\end{pmatrix} \mapsto \begin{pmatrix}
-S^x \\
~~S^y \\
-S^z
\end{pmatrix}, ~
P_z:\begin{pmatrix}
S^x \\
S^y \\
S^z
\end{pmatrix} \mapsto \begin{pmatrix}
-S^x \\
-S^y \\
~~S^z
\end{pmatrix}
\end{equation}
and $P_x$ is as defined before in (\ref{eq:Px}). A model that retains only this symmetry is the XYZ spin chain 
\begin{equation}
\label{eq:XYZ}
H =  \sum_i \left( J_x S^x_i S^x_{i+1} + J_y  S^y_i S^y_{i+1} + J_z S^z_i S^z_{i+1}\right) 
\end{equation} 
which,  depending on relative strengths of $J_x,J_y,J_z$ contains three distinct N\'{e}el ordered phases corresponding to $\ztzt$ being broken down to one of its three $\ztwo$ subgroups generated by one of $P_x,P_y,P_z$. Since the three $\ztwo$ groups are not subgroups of each other, these N\'{e}el-to-N\'{e}el phase transitions are also Landau-forbidden. Recently,  Mudry et al~\cite{Mudry_DQC_PhysRevB.99.205153} investigated a next-nearest-neighbor deformation of the XYZ model of eq~(\ref{eq:XYZ}) and again found a rich phase diagram including a N\'{e}el-to-VBS transition.

The modern understanding that connects DQC and LSM is that both Landau forbidden transitions as well as the absence of a trivial phase is a consequence of the kinematics of the physical systems - the presence of an anomaly~\cite{ElseThorngren_AnomalyLSM2019topological,Cho_AnomalyLSM_PhysRevB.96.195105,Metlitski_Thorngren_DQCAnomalies_PhysRevB.98.085140}. While anomalies in quantum field theories manifest themselves as an obstruction to \emph{gauging}, in lattice models, they present themselves by obstructing a trivial phase and allowing Landau-forbidden transitions. Another context where anomalies are present are on the boundaries of symmetry protected topological (SPT) phases where it was found that when the system belongs to a non-trivial SPT phase, in the presence of boundaries, the ground state can never be unique and gapped. Thus, LSM, DQC and SPT phenomena are related by anomalies. 

\section{Summary}
We have provided an elementary proof of a class of LSM theorems in 1D using the matrix-product-state representation of spin chains. We prove that the presence of on-site projective symmetries and translation invariance forbids the existence of a trivial phase where the spin chain has a unique, gapped ground state. 

The generalization of on-site unitary symmetries to also include time-reversal and reflections should be a straight forward and interesting extension of this work. Even more interesting is the extension of these results to fermionic systems where the notion of a trivial phase is subtle and the presence of fermion parity would allow more complex versions of projective representations and LSM theorems. We leave these directions for future work.

\section{Note added}
After the completion of this work and its appearance on arXiv, I was informed of refs~(\onlinecite{ChenGuWen_old_PhysRevB.83.035107,SchuchGarciaCirac_PhysRevB.84.165139})~(see also the discussion on Physics StackExchange on \href{https://bit.ly/32NkFxa}{bit.ly/32NkFxa}) where the main idea of proving the LSM constraint using matrix product states was already discussed. However some novel technical aspects like the Wigner Eckart theorem as well as substantial details are presented in this paper which are unavailable in the above references. I am grateful to Ruben Verrsen for bringing this to my attention.

\section{Acknowledgments}
I acknowledge useful discussions with Subhro Bhattacharjee, Ashvin Vishwanath, Masaki Oshikawa, Diptiman Sen, Ganapathi Bhaskaran, Vijay Shenoy and especially Hal Tasaki whose talk brought my attention to this problem. I am grateful to Tzu-Chieh Wei for a careful reading of the manuscript, his suggestions and pointing out an error. This work was initiated and completed during the program `Thermalization, Many Body Localization and Hydrodynamics' (ICTS/hydrodynamics2019/11) at ICTS-TIFR and partially drafted during the programs `Novel Phases of Quantum Matter' (ICTS/topmatter2020/01) at ICTS-TIFR and `Gapless fermions: from Fermi liquids to strange metals' at Max Planck Institute for the Physics of Complex Systems. I also acknowledge funding from the Simon's foundation through the ICTS-Simons postdoctoral fellowship. 

\newpage 
\bibliography{references}{}

\begin{thebibliography}{33}
\bibitem [{Tzu()}]{Tzu-Chieh}%
  \BibitemOpen
  \href@noop {} {}\bibinfo {note} {A much simpler way to see the lack of
  injectivity is by observing that, when the MPS matrices are block-diagonal $
  C^m = \begin{pmatrix} A^m & 0 \\ 0 & B^m \end{pmatrix},~D = \begin{pmatrix}
  B^m & 0 \\ 0 & A^m \end{pmatrix}$, for any general $M = \begin{pmatrix}
  M_{11} & M_{12} \\ M_{21} & M_{22} \end{pmatrix}$, the off-diagonal blocks do
  not contribute to the evaluation of $\Gamma_I(M)$ and thus any choice of
  $M_{12}$ and $M_{21}$ gives the same $\Gamma_I(M)$ leading to a loss of
  injectivity. This generalizes to an arbitrary number of blocks in a
  straight-forward manner. I thank Tzu-Chieh Wei for pointing this out. For a
  fixed $M$ as we considered in the main text, $\Theta M$ and $M \Theta$ differ
  precisely in the off-diagonal blocks and thus we have $\Gamma_I(\Theta M) =
  \Gamma_I(M \Theta)$.}\BibitemShut {Stop}%
\bibitem [{irr()}]{irrep_note}%
  \BibitemOpen
  \href@noop {} {}\bibinfo {note} {In principle, we could consider the physical
  spin to contain multiple copies of multiple irreps. This does not change much
  and merely requires us to introduce multiple labels to keep track of the
  states of spin.}\BibitemShut {Stop}%
\bibitem [{SSB()}]{SSB_note}%
  \BibitemOpen
  \href@noop {} {}\bibinfo {note} {Although in most practical situations it is
  clear from context, it is usually not meaningful to say ``symmetry $K$ is
  broken" rather it is meaningful to say ``symmetry $H$ is left unbroken". In
  other words, given a system with global symmetry $G$ that is ordering via
  SSB, we can accurately specify the pattern of symmetry breaking by stating
  the subgroup $H \subset G$ that acts as a residual symmetry after ordering.
  The GSD is $|G/H|$. When we say ``the generator of a finite symmetry is
  broken", we mean $H$ is such that $G/H$ is finite. Note that unless $H$ is a
  normal subgroup, the coset $G/H$ is not itself a group.}\BibitemShut {Stop}%
\end{thebibliography}

\newpage

\appendix
\section{Proof of intermediate lemmas used in the main text}
\label{app:proofs}

\begin{lemma}
	\label{lemma:SO(3)_Bond_hilbert_space}
	The bond Hilbert space of an injective MPS invariant under spin rotations and lattice translations has either only integer or only half-odd-integer spins.
\end{lemma}
\begin{proof}
	Let us prove this by assuming the contrary i.e. the MPS is injective and the bond Hilbert space contains both integer and half-odd-integer irreps. Consider the spin $J$ rotation matrix for $2 \pi$ rotation about any axis, $\hat{n}$.
	\begin{equation}
		D(\hat{n}, 2 \pi) = (-1)^{2j} \mathbb{1}
	\end{equation}
	 Substituting in eq~(\ref{eq:SO(3)_MPS_onsite_translation_invariance}, this gives us
	\begin{equation}
	\label{eq:2 pi rotation }
	(-1)^{2J} A^m = X A^m X
	\end{equation}
	where, $J$ is the physical spin which can be integer or half-odd-integer, and 
	\begin{equation}
	X = V(\hat{n}, 2 \pi) = \begin{pmatrix}
	\mathbb{1} & 0 \\
	0 & -\mathbb{1}
	\end{pmatrix}.
	\end{equation}
	We have chosen the basis of the bond space such that the top block of $V(g$) corresponds to integer $j$ and the bottom corresponds to half-odd-integer $j$ and $X^2 = \mathbb{1}$. Now, observe that
	\begin{multline}
	\Gamma_I(X M X) \\ = \sum_{m_1 \ldots m_I} Tr(X M X A^{m_1} \ldots A^{m_I}) \ket{m_1, m_2, \ldots m_I} \\
	= \sum_{m_1 \ldots m_I} Tr( M X A^{m_1}X \ldots XA^{m_I}X) \ket{m_1, m_2, \ldots m_I} \\
	= (-1)^{2JI} \sum_{m_1 \ldots m_I} Tr( M A^{m_1} \ldots A^{m_I}) \ket{m_1, m_2, \ldots m_I} \\=
	(-1)^{2JI} \Gamma_I(M) 
	\end{multline}
	Thus, for any length $I$, we can always find two distinct matrices $M$ and $N=(-1)^{2JI} X M X$, such that $ \Gamma_I(M) = \Gamma_I(N)$ thus violating injectivity and contradicting the original assumption.
\end{proof}

\begin{lemma}
	\label{lemma:G_Bond_hilbert_space}
	The bond Hilbert space of an injective MPS invariant under on-site $G$ symmetry and lattice translations contains projective irreducible representations of a single class.
\end{lemma}
\begin{proof}
	Let us prove this by assuming the contrary i.e. the MPS is injective and the bond Hilbert space contains projective irreps of various classes corresponding to elements of the group $\omega \in \hgc$. Let us consider the following transformation of the physical spin (corresponding to projective irrep $J$ of class $\omega_J \in \hgc$) 
	\begin{equation}
	D(g)D(h)D^{\dagger}(gh) = \omega_J(g,h) \mathbb{1}
	\end{equation}
Substituting in eq~(\ref{eq:MPS_onsite_translation_invariance}), this gives us
\begin{equation}
\omega_J(g,h) A^m = X^\dagger A^m X 
\end{equation}
where
\begin{equation}
X = V(g)V(g)V^\dagger(gh) 
= \bigoplus_{[\omega] \in \hgc} \omega(g,h) \mathbb{1}.
\end{equation}
We have chosen the basis for the bond Hilbert space so that it is block-diagonal with each class of projective irreps grouped together into blocks as shown above. Observe that
\begin{multline}
\Gamma_I(X M X^\dagger) \\= \sum_{m_1 \ldots m_I} Tr(X M X^\dagger A^{m_1} \ldots A^{m_I}) \ket{m_1, m_2, \ldots m_I} \\
= \sum_{m_1 \ldots m_I} Tr( M X^\dagger A^{m_1}X \ldots X^\dagger A^{m_I}X) \ket{m_1, m_2, \ldots m_I} \\
= (\omega_J(g,h))^I \sum_{m_1 \ldots m_I} Tr( M A^{m_1} \ldots A^{m_I}) \ket{m_1, m_2, \ldots m_I} \\=
(\omega_J(g,h))^I \Gamma_I(M) 
\end{multline}
Thus, for any length $I$, by appropriately choosing $g$ and $h$(such that $\omega(g,h)$ is non-trivial following, say, the prescription listed in the appendix of the paper by by Pollmann and Turner~\cite{PollmannTurner_PhysRevB.86.125441}), we can always find two distinct matrices $M$ and $N=(\omega^*_J(g,h))^I X M X^\dagger$, such that $ \Gamma_I(M) = \Gamma_I(N)$ thus violating injectivity and contradicting the original assumption.
\end{proof}

\section{Some comments on projective representations}
\label{app:projective representations}
Given a group G, a linear representation is an assignment of invertible, complex matrices $M(g)$ for each element $g \in G$ compatible with group multiplication i.e. $M(g) M(h) = M(gh)$. A projective representation on the other hand is an assignment of invertible, complex matrices $V(g)$ for each element $g \in G$ such that group multiplication holds upto on overall complex phase $\omega$
\begin{equation}
    V(g) V(h) = \omega(g,h) V(g.h).
\end{equation}
Associativity of $g \in G$ constrains the $\omega$s to follow the following \emph{cocycle} condition
\begin{equation}
	\label{eq:2cocycle}
    \omega(g,h) \omega(gh,l) = \omega(g,hl) \omega(h,l)
\end{equation}
Projective representations fall into equivalence classes where the equivalence relation comes from the effect of rephasing $V(g)$ 
\begin{equation}
V(g) \sim  \bar{V}(g) = \theta_g V(g) 
\end{equation}
which, using eq~(\ref{eq:2cocycle}) gives us 
\begin{equation}
\label{eq:coboundary}
\omega(g,h) \sim \omega(g,h) \frac{\theta_g \theta_h}{\theta_{gh}}
\end{equation}
The equivalence classes of solutions to the set of equations (\ref{eq:2cocycle}) subject to the equivalence relation of eq~(\ref{eq:coboundary}) correspond to equivalence classes of projective representations. These classes have an Abelian group structure under multiplication- the second group-cohomology group, $\hgc$.

Let us now understand a fact stated in the main text- \emph{given a projective representation of $G$, there exists a subgroup $H$ of $G$ such that when restricted to $H$, the representation is linear to $H$}. That there always exists such a group is easy to establish- the trivial group has no projective representations and is always a subgroup of any $G$. However, it is also interesting to understand the condition for the subgroup to be non-trivial. An elegant way to understand this is the following- $H$ is a subgroup of $G$ implies there exists an \emph{injective} group homomorphism,
\begin{equation}
i:H \hookrightarrow G.
\end{equation}
As explained above, a projective representation of $G$ is represented by some collection of 2-cocycles $\omega(g_1,g_2)$ such that $[\omega] \in \hgc$. Now, given the group homomorphism above, we can obtain a collection of 2-cocycles of $H$ by \emph{pullback}
\begin{equation}
i^*\omega(h_1,h_2) = \omega(i(h_1),i(h_2)).
\end{equation}
In other words, the homomorphism $i$ induces the following group homomorphism
\begin{equation}
\label{eq:pullback}
i^*: \hgc \rightarrow \hgcinput{H}.
\end{equation}
The $H$ we seek is such that the image of $[\omega] \in \hgc$ is the trivial element in the map of eq~(\ref{eq:pullback}). It is clear that $H$ being the trivial group works trivially. Let us look at one example. When $G$ is $SO(3)$, $\hgc \cong \ztwo$. If we consider $H$ to be either $O(2)$ or $\ztzt$, $\hgcinput{H}\cong \ztwo$ and $i^*$ merely induces the identity map
\begin{equation}
i^*:\ztwo \hookrightarrow \ztwo.
\end{equation}
However, if $H$ is $SO(2)$, $\ztwo$ or the trivial group $1$, $\hgcinput{H} \cong 1$ and $i^*$ induces the trivial map 
\begin{equation}
i^*: \ztwo \rightarrow 1
\end{equation} 
Thus, restricting the projective (half-odd-integer) irreps to either $O(2)$ or $\ztzt$ subgroups still leaves them a as projective representations of the subgroups but restricting them to $SO(2)$, $\ztwo$ or the trivial group leaves them as linear representations.

\end{document}